\documentclass[onecolumn, draftcls,12pt]{IEEEtran}

\usepackage{amssymb,amsmath,amsfonts,amsthm}
\usepackage{algorithm,algorithmic}
\usepackage{boxedminipage}
\usepackage{cite}
\usepackage[dvips]{graphicx}
\usepackage{epsfig}
\usepackage{float}
\graphicspath{{figures/}}
\usepackage{graphicx}
\usepackage{graphics}
\usepackage{multirow}
\usepackage{threeparttable} 
\usepackage[usenames,dvipsnames]{color}

\usepackage{verbatim}

\newtheorem{theorem}{Theorem}
\newtheorem{lemma}{Lemma}

\newtheorem{remark}{Remark}

\newtheorem{proposition}{Proposition}

\newtheorem{problem}{Problem}

\newcommand{\sr}{\stackrel}

\newcommand{\rar}{\rightarrow}

\newcommand{\tri}{\sr{\triangle}{=}}

\newcommand{\be}{\begin{equation}}
\newcommand{\ee}{\end{equation}}
\newcommand{\bea}{\begin{eqnarray}}
\newcommand{\eea}{\end{eqnarray}}
\newcommand{\bes}{\begin{eqnarray*}}
\newcommand{\ees}{\end{eqnarray*}}

\newcommand{\bi}{\begin{itemize}}
\newcommand{\ei}{\end{itemize}}
\newcommand{\ben}{\begin{enumerate}}
\newcommand{\een}{\end{enumerate}}

\newcommand{\bp}{\begin{problem}}
\newcommand{\ep}{\end{problem}}
\newcommand{\hso}{\hspace{.1in}}
\newcommand{\hst}{\hspace{.2in}}

\newcommand{\noi}{\noindent}

\newcommand{\bc}{\begin{center}}
\newcommand{\ec}{\end{center}}

\floatstyle{ruled}
\newfloat{algorithm}{H}{algo}
\floatname{algorithm}{\footnotesize Algorithm}

\begin{document}
%
\title{Optimal Merging Algorithms for Lossless Codes with Generalized Criteria\footnote{Parts of preliminary results of this work have been published by the authors in \cite{2011:charalambousISIT}.} }


\author{ Themistoklis Charalambous,~\IEEEmembership{Member,~IEEE,} Charalambos D. Charalambous,~\IEEEmembership{Senior Member,~IEEE,} and Farzad Rezaei,~\IEEEmembership{Member,~IEEE}
\thanks{T. Charalambous and C.D. Charalambous are with the Department of Electrical and Computer Engineering, University of Cyprus, Nicosia 1678 (E-mail: \{themis, chadcha\}@ucy.ac.cy).}
\thanks{Farzad Rezaei is an alumni of the School of Information Technology and Engineering, University Ottawa, Ottawa, Canada (E-mail:  frezaei@alumni.uottawa.ca).}
}

\maketitle

%
%
%
%
\begin{abstract}
  This paper presents lossless prefix codes optimized with respect to a pay-off criterion consisting of a convex combination of maximum codeword length and average codeword length. The optimal codeword lengths obtained are based on a new coding algorithm which transforms the  initial source probability vector into a new probability vector according to a merging rule. The coding algorithm is equivalent to a partition of the source alphabet into disjoint sets on which a new transformed probability vector is defined as a function of the initial source probability vector and a scalar parameter. The pay-off criterion considered encompasses a trade-off between maximum and average codeword length; it is related to a  pay-off criterion consisting of a convex combination of average codeword length and average of an exponential function of the codeword length, and to an average codeword length pay-off criterion subject to a limited length constraint. A special case of the first related pay-off is connected to coding problems involving source probability uncertainty and codeword overflow probability, while the second related pay-off compliments limited length Huffman coding algorithms.
\end{abstract}

\IEEEpeerreviewmaketitle

%
%
%
%
\section{Introduction}\label{intro}


\noi Lossless fixed to variable length source codes are usually examined under known source  probability distributions, and unknown source probability distributions. For known source probability distributions there is an extensive literature which aims at minimizing various pay-offs such as the average  codeword length \cite[Section 5.3]{2006:Cover}, the average redundancy of the codeword length \cite{2004:DrmotaSzpankowski,2008a:Baer}, the average of an exponential function of the codeword length    \cite{1965:campbell_coding,1981:humblet_generalization,2008b:Baer}, the average of an exponential function of the redundancy of the codeword length \cite{2006a:Baer,2008a:Baer, 2008b:Baer}, and the probability of codeword length overflow \cite{1967:Jelinek,1991:Merhav}.  On the other hand, universal coding and universal modeling, and the so-called Minimum Description Length (MDL) principle  are often examined via minimax techniques, when the source probability distribution is unknown, but belongs to a pre-specified  class of source distributions  \cite{1973:davisson,1980:davisson_Leon-Garcia, 2004:DrmotaSzpankowski,2003:JacquetSzpankowski,2009:Gawrychowski_Gagie}. With respect to the above pay-offs Shannon codes find  sub-optimal code lengths by treating them as real numbers, while Huffman codes  find the optimal  code lengths by treating them as integers. Coding algorithms for general pay-off criteria involving pointwise redundancy, average exponential redundancy, and maximum pointwise redundancy are found in \cite{20011:Baer}.


\noi The main objectives of this paper are to introduce a new pay-off criterion consisting of a convex combination of the maximum codeword and average codeword length,  to derive lossless prefix codes,  to discuss the implication of these codes to variable length coding applications, and to identify relations of the new pay-off to other pay-offs addressed in the literature. The criterion considered incorporates a trade-off between average codeword length and maximum codeword length, which makes the new coding algorithm  suitable for length sensitive coding applications. It  is general enough to  encompass as a special case some of the pay-off criteria investigated in the literature, such as limited-length coding \cite{2010:golinZhang} and  coding for exponential functions of the codeword length, while it is easily generalized to  universal coding in which the source probability vector belongs to a class.  \\
The new pay-off criterion  considered is discussed under  Problem~\ref{problem1} of Section~\ref{sec:statement},  while its  connections to other pay-off criteria such as limited-length codes and codes obtained via convex combination of average and exponential function of the codeword length are discussed in Sections~\ref{limitedLength} and~\ref{limit}, respectively. Extensions of the new pay-off to universal codes is discussed in Section~\ref{sec:universal}.

%
%
%
%
\subsection{Problem Formulation and Discussion of Results}\label{sec:statement}

\noi Consider a source with alphabet  ${\cal X } \tri \{x_1, x_2, \ldots, x_{ | {\cal X }| }  \}$  of cardinality   $|{\cal X}|$, generating symbols according to  the probability distribution ${\bf p} \tri \{p(x): x \in {\cal X}  \} \equiv \left(p(x_1), p(x_2), \ldots, p(x_{|{\cal X}|})\right)$. Source symbols are encoded into $D-$ary codewords. A code ${\cal C} \tri \{c(x): x \in {\cal X}\}$ for symbols in ${\cal X}$ with image alphabet ${\cal D} \tri \{0, 1, 2, \ldots, D-1 \}$ is an injective map    $c: {\cal X} \rar {\cal D}^*$, where ${\cal D}^*$ is the set of finite sequences drawn from ${\cal D}$.  For $x \in {\cal X}$  each codeword $c(x) \in {\cal D}^*,  c \in {\cal C}$ is identified with a codeword length $l(x) \in {\mathbb Z}_+$, where ${\mathbb Z}_+$ is the set of non-negative integers. Thus, a  code ${\cal C}$ for source symbols from the alphabet ${\cal X}$ is associated with the length function of the code $ l : {\cal X} \rar {\mathbb Z}_+$, and a code defines a codeword length vector ${\bf l} \tri \{ l(x): x \in {\cal X}\} \equiv \big(l(x_1), l(x_2), \ldots, l(x_{|{\cal X}|})\big) \in {\mathbb Z}_+^{|{\cal X}|}$.
Since a function $l : {\cal X} \rar {\mathbb Z}_+$ is the length function of some prefix code if, and only if, the Kraft inequality holds \cite[Section 5.2]{2006:Cover}, then the admissible set of codeword length vectors is defined by
${\cal L}\left( {\mathbb Z}_+^{|{\cal X}|}  \right)  \tri \Big\{ {\bf l} \in {\mathbb Z}_+^{|{\cal X}|} : \sum_{x \in {\cal X}} D^{-l(x)} \leq 1 \Big\}.$
On the other hand, if the integer constraint is relaxed by admitting real-valued length vectors ${\bf l} \in {\mathbb R}_+^{|{\cal X}|}$, which satisfy the Kraft inequality, such as Shannon codes or arithmetic codes, then ${\cal L}\left( {\mathbb Z}_+^{|{\cal X}|} \right)$ is replaced by
\bes
{\cal L}\left( {\mathbb R}_+^{|{\cal X}|} \right)  \tri \Big\{ {\bf l} \in {\mathbb R}_+^{|{\cal X}|} : \sum_{x \in {\cal X}} D^{-l(x)} \leq 1 \Big\}.
\ees
Such codes are useful in obtaining approximate solutions which are less computationally intensive \cite[Section 5.3]{2006:Cover}.
Without loss of generality, it is assumed that the set of probability distributions is defined by
\begin{align*}
{\mathbb P}({\cal X})& \tri \Big\{ {\bf p} =\Big( p(x_1), \ldots, p(x_{|{\cal X}|})\Big) \in {\mathbb R}_+^{|{\cal X}|} : \\
&p(x_{|{\cal X}|})>0, \hso p(x_i) \leq p(x_j), \forall i>j,  (x_i,x_j) \in {\cal X},  \sum_{x \in {\cal X}} p(x) =1 \Big\}.
\end{align*}

\noi Unless specified otherwise, the following notation is used:  $\log (\cdot) \tri \log_D (\cdot)$, and ${\mathbb H}({\bf p})$ is the entropy of the probability distribution ${\mathbf p}$.

\noi The main  pay-off considered  is  a convex combination of the maximum codeword length and the average codeword length.  Specifically,  a parameter $\alpha \in [0,1]$ is introduced which weights the maximum codeword length, while $(1-\alpha)$ weights the average codeword length, and as this parameter moves away from $\alpha=0$, more weight is put on reducing the maximum codeword length, thus the maximum length of the code is reduced resulting in a more balanced code tree. Such pay-off is particularly important in applications where the codeword lengths are bounded by a specific constant. The main  problem investigated  is stated below.

\begin{problem}\label{problem1}
Given a known source probability vector ${\bf p} \in   {\mathbb P}({\cal X})$ and weighting parameter $\alpha \in [0,1] $, find a prefix code length vector ${\bf l}^* \in {\mathbb R}_+^{|{\cal X}|}$ which minimizes the Maximum and Average Length pay-off ${\mathbb L}_{\alpha}({\bf l}, {\bf p})$ defined by
\bea
{\mathbb L}_{\alpha}({\bf l}, {\bf p}) \tri \Big\{ \alpha  \| {\bf l} \|_\infty+ (1-\alpha) \sum_{ x \in {\cal X}} l(x) p(x)\Big\}, \quad \| {\bf l} \|_\infty \tri \max_{ x\in {\cal X}} l(x).  \label{b30}
\eea
\end{problem}

\noi The presence of the $\ell_\infty$ norm (e.g., $||l||_{\infty}$) in the pay-off ${\mathbb L}_{\alpha}({\bf l}, {\bf p})$ makes the characterization of the optimal real-valued prefix code, which is parametrically dependent on $\alpha \in [0,1]$,  very different from previously known Shannon type codes. Indeed, it is shown in subsequent sections that  the optimal code corresponding to Problem~\ref{problem1}  is equivalent to a specific partition of the source alphabet, and re-normalization  and merging of entries of the initial source probability vector,  as a function  of the parameter $\alpha \in [0,1]$, from which the optimal code is derived.  The single letter performance of the optimal codeword lengths  $\{l^\dagger(x): x \in {\cal X}\}$ satisfy ${\mathbb H}({\bf w}_\alpha) \leq {\mathbb L}_{\alpha}({\bf l}^\dagger, {\bf p})
< {\mathbb H}({\bf w}_\alpha) +1$, where ${\bf w}_\alpha \tri \{ w_\alpha(x): x \in {\cal X}\}$ is a new probability vector which depends on the initial source probability vector and the parameter $\alpha \in [0,1]$. As $\alpha \in [0,1]$ increases the optimal code tree moves towards the direction of a more balanced code tree while there is an $\alpha_{max} \in [0,1]$ which is the minimum value beyond which there is no compression. \\
An algorithm is presented which computes the weight vector ${\bf w}_\alpha$ via partitioning of the source alphabet,  re-normalization and merging of the initial source probability vector, for any value of $\alpha \in [0,1]$, having a worst case computational complexity of order $\mathcal{O}(n)$. \\
Problem \ref{problem1}, as suggested by one of the reviewers, can also be solved in a waterfilling-like fashion. For completeness and direct comparison with the methodology suggested in this paper, the solution to Problem \ref{problem1} is included in Appendix \ref{AppendixB}.

\subsection{Relations to Literature}
In Section~\ref{limitedLength} it is shown that  limited-length coding problems defined  by  minimizing the average codeword length subject to a maximum codeword length constraint (Problem~\ref{problem3a}) are deduced from the solution of Problem \ref{problem1} as a special case. This connection provides Shannon type codes, and compliments the recent work on limited-length Huffman codes \cite{2010:golinZhang}. Specifically, given a hard constraint $L_{\lim} \in [1,\infty)$, the problem of finding a prefix code length vector ${\bf l}^* \in {\mathbb R}_+^{|{\cal X}|}$ which minimizes the Average Length Subject to Maximum Length Constraint pay-off ${\mathbb L}({\bf l}, {\bf p})$ defined by
\begin{subequations}\label{eq:problem3a}
\begin{align}
& \mathbb{ L}({\bf l}, {\bf p})  \tri \sum_{x\in {\cal X}} l(x) p(x),  \label{j3a}\\
& \text{subject to} \quad  \max_{x\in {\cal X}} l(x) \leq L_{\lim}. \label{j3b}
\end{align}
\end{subequations}
for all $ \alpha \in [0,1]$ is obtained from the solution of Problem \ref{problem1}. The connection is established  by introducing a real-valued Lagrange multiplier $\mu$ associated with the constraint on the maximum length, via  the unconstrained pay-off  defined by
\begin{align} \label{Lag_1}
 {\mathbb L}({\bf l}, {\bf p}, \mu) & \tri \sum_{x \in {\cal X}} l(x) p(x) +  \mu(\max_{x\in {\cal X}} l(x)- L_{\lim}), \quad \mu>0 \nonumber \\
 & =\mu \max_{x\in {\cal X}} l(x) + \sum_{x \in {\cal X}} l(x) p(x) -\mu L_{\lim}.
\end{align}
Hence, the optimal code for limited-length codes is obtained from the optimal code solution of  Problem 1, by substituting $\mu=\alpha/(1-\alpha)$, and then relating the value of the  Lagrange multiplier with a specific value of $\alpha$ for which the codeword lengths will be limited by $L_{\lim}$. The complete characterization of the solution to such problems is given in  Section~\ref{limitedLength}, which also includes an algorithm.\\

\noi In Section~\ref{limit} it is shown that Problem~\ref{problem1} is also related to a general-pay off consisting of a convex combination of the average codeword length  and average of an exponential function of codeword length  (Problem~\ref{problem2}) defined by
 \begin{align}
{\mathbb L}_{t,\alpha}({\bf l}, {\bf p}) &  \tri \frac{\alpha}{t} \log \Big(\sum_{ x \in {\cal X}} p(x) D^{t \l(x)}\Big) + (1-\alpha) \sum_{ x \in {\cal X}} l(x) p(x), \label{gcc}
\end{align}
 where $t \in (-\infty,\infty)$ is another parameter. Specifically, by noticing that  $\frac{1}{t}\log \sum_{ x \in {\cal X}} p(x) D^{t l(x)}$ is a nondecreasing function of $t \in [0,\infty)$, and $\lim_{t \rar \infty} \frac{1}{t}\log \sum_{ x \in {\cal X}} p(x) D^{t l(x)}=\max_{ x \in {\cal X}} l(x)$, then by replacing  $\alpha \max_{ x \in {\cal X}} l(x)$ in ${\mathbb L}_{\alpha}({\bf l}, {\bf p})$,  by the function $\frac{\alpha}{t} \log \Big(\sum_{ x \in {\cal X}} p(x) D^{t \l(x)}\Big)$, the resulting pay-off takes into account moderate values below  $\max_{ x \in {\cal X}} l(x)$, obtaining a two-parameter pay-off (\ref{gcc}). The pay-off ${\mathbb L}_{t,\alpha}({\bf l}, {\bf p})$ is  a convex combination of the average of an  exponential function of the codeword length, and the average codeword length. The case $\alpha=1$ is investigated in \cite{1965:campbell_coding,1981:humblet_generalization,2006a:Baer,2008a:Baer, 2008b:Baer,1991:Merhav}, where relations to minimizing  buffer overflow probability are discussed. Further, it is not difficult to verify that ${\mathbb L}_{t,\alpha}|_{\alpha=1}({\bf l}, {\bf p})$ is also the dual problem of universal coding problems, formulated as a minimax, in which the maximization is over  a class of probability distributions which satisfy a relative entropy constraint with respect to a given fixed nominal probability distribution  \cite{2005:rezaei_bambos,2009:Gawrychowski_Gagie}. Hence, the pay-off ${\mathbb L}_{t,\alpha}|_{\alpha=1}({\bf l}, {\bf p})$ encompasses a trade-off between universal codes and buffer overflow probability  and average codeword length codes.  Since the pay-off ${\mathbb L}_{t,\alpha}({\bf l}, {\bf p})$ is  in the limit, as $t \rar \infty$, equivalent to $\lim_{t \rar \infty} {\mathbb L}_{t,\alpha}({\bf l}, {\bf p})  ={\mathbb L}_{\alpha}({\bf l}, {\bf p})$, $\forall \alpha \in [0,1]$, then the codeword length vector minimizing ${\mathbb L}_{t,\alpha}({\bf l}, {\bf p})$ is expected to converge in the limit as $t \rar \infty$, to that which minimizes ${\mathbb L}_{\alpha}({\bf l}, {\bf p})$. However, moderate values of $t \in [0,\infty)$ are also of interest since the pay-off ${\mathbb L}_{t,\alpha}({\bf l}, {\bf p})$ can be interpreted as a trade-off between universal codes and average length codes.\\

 \noi Finally, in Section~\ref{sec:universal} it is demonstrated that the optimal codes obtained for the new pay-off can be used to solve universal coding problems, formulated as minimax problems,  in which the initial source probability vectors belongs to a specified family of probability vectors ${\mathbb S}({\cal X}) \subset {\mathbb P}({\cal X})$, with respect to the pay-off
 \bea
 {\mathbb L}_{\alpha}^+({\bf l}, {\bf p}) \tri  \max_{  {\bf p} \in   {\mathbb S}({\cal X})}  \Big\{ \alpha \max_{ x\in {\cal X}} l(x) + (1-\alpha) \sum_{ x \in {\cal X}} l(x) p(x)\Big\}. \label{uc}
\eea



\noi The rest of the paper is organized as follows. Section~\ref{sec:weights} addresses Problem~\ref{problem1} and derives basic results concerning the partition of the source alphabet, the re-normalization and merging rule as $\alpha$ ranges over $[0,1]$. Here, an algorithm is presented which describes how the partition of the source alphabet is characterized. Section~\ref{sec:lengths} gives the complete characterization of optimal codes corresponding to Problem~\ref{problem1}, the associated coding theorem, and relations to limited-length coding problems (Problem~\ref{problem3a}), and coding problems with general-pay off consisting of a convex combination of the average codeword length and average of an exponential function of codeword length  (Problem~\ref{problem2}).  Section~\ref{sec:examples} provides illustrative examples. Finally, Section~\ref{sec:conclusions} presents the conclusions and identifies open problems for future research.

%
%
%
%

\section{ Optimal Weights and Merging Rule}\label{sec:weights}

The main objective of this section is to convert the pay-off  of Problem~\ref{problem1}  into an equivalent  objective of the form $\sum_{x \in {\cal X}} w_\alpha(x) l(x) $,  where the new weights  ${\bf w}_\alpha \tri \{w_\alpha(x): x \in {\cal X}\}$  depend parametrically on $\alpha \in [0,1]$. Subsequently, we derive certain properties of the new weight vector as a function of the initial source probability vector and $\alpha \in [0,1]$,  and identify how these properties are transformed into equivalent properties for the optimal codeword length vector. The main issue here is to identify how  symbols are merged  together, and how the merging  changes as a function of the parameter $\alpha \in [0,1]$ and initial source probability vector, so that the optimal solution is characterized for all $\alpha \in [0,1]$. From these properties the optimal real-valued  codeword lengths for Problem~\ref{problem1} will be found. This merging will also provide insight in  characterizing  optimal codes for related problems (with different pay-offs).

\noi Define
\bea
\displaystyle l^* \tri \max_{ x \in {\cal X}} l(x), \hst \mathcal{U} \tri \Big\{x \in {\cal X}: l(x) = l^* \Big\} \label{n1} .
\eea
\noi The pay-off ${\mathbb L}_\alpha({\bf l}, {\bf p})$ can be written as
\begin{align}
 {\mathbb L}_\alpha({\bf l}, {\bf p}) & = \alpha l^*+(1-\alpha) \sum_{x \in {\cal X}} l(x) p(x) =\Big(\alpha+(1-\alpha)\sum_{x \in \mathcal{U}} p(x) \Big) l^*+ \sum_{x \in {\cal U}^c }(1-\alpha) p(x) l(x),  \label{n2}
 \end{align}
which makes the dependence on the disjoint sets ${\cal U}$ and  ${\cal U}^c \tri {\cal X} \setminus{\cal U}$ explicit. The set ${\cal U}$ remains to be identified so that a solution to the coding problem exists for all $\alpha \in [0,1]$.\\
Note that $l^* \equiv l^*(\alpha)$ and ${\cal U}\equiv {\cal U}(\alpha)$, that is, both the maximum length and the set of source symbols which correspond to the maximum length depend   parametrically on $\alpha \in [0,1]$. This explicit dependence will often be omitted for simplicity of notation. \\
Define
\begin{align}
\sum_{x \in {\mathcal U} }  w_{\alpha}(x)= \Big(\alpha+(1-\alpha)\sum_{x \in \mathcal{U}} p(x) \Big), \hst
w_{\alpha}(x)&=(1-\alpha) p(x),~ x \in {\cal U}^c.  \label{n3}
 \end{align}
Using (\ref{n2}) and (\ref{n3}) the pay-off ${\mathbb L}_\alpha({\bf l}, {\bf p})$ is written as a function of the new weight vector as follows:
\bea
{\mathbb L}_\alpha({\bf l}, {\bf p})\equiv {\mathbb L}({\bf l}, {\bf w}_\alpha) \tri \sum_{x \in {\cal X}}^{} w_\alpha(x) l(x), \hst \forall \alpha \in [0,1]. \label{n4}
\eea
The new weight vector ${\bf w}_{\alpha}$ is a function of $\alpha$ and the source probability vector ${\bf p} \in {\mathbb P}({\cal X})$, and it is defined over the two disjoint sets ${\cal U}$ and  ${\cal U}^c$. It can be easily verified that $0\leq w_{\alpha}(x) \leq 1,~\forall x \in {\cal U}^c$ and $\sum_{ x \in {\cal X} }w_{\alpha}(x) =1,  \forall \alpha \in [0,1]$. However, at this stage it cannot be verified that  $w_\alpha(x)\geq 0, \forall x \in {\cal U}$.

\noi The next lemma finds the optimal codeword length vector.

\begin{lemma}\label{lemma_prelimin}
 The real-valued prefix codes minimizing pay-off ${\mathbb L}_\alpha({\bf l}, {\bf p})$ for $\alpha \in [0,1)$ are given by
\bea \label{eq:solution1}
l^\dagger(x) = \left\{ \begin{array}{ll}
-\log{\Big((1-\alpha)p(x)\Big)}=-\log w_\alpha^\dagger(x),  &  x \in {\cal U}^c  \\
-\log{\Big( \frac{\alpha \sum_{ x \in {\cal U}^c }p(x) + \sum_{x \in \mathcal{U}}p(x)}{  |\mathcal{U}|} \Big)}=-\log w_\alpha^\dagger(x), &  x \in {\cal U} \end{array} \right.
\eea
where ${\cal U}$ and ${\cal U}^c$  remain to be identified. Note that for $\alpha =1$, ${\cal X}= {\cal U}$ and $l^\dagger(x)=\log |{\cal X}|, \forall x \in {\cal X}$.
\end{lemma}

\begin{proof}
See Appendix \ref{proof_lemma_prelimin}.
\end{proof}

\noi The point to be made regarding Lemma~\ref{lemma_prelimin} is twofold. Firstly, since for $\alpha \in [0,1)$ the pay-off ${\mathbb L}_\alpha({\bf l}, {\bf p})$ is continuous in ${\bf l}$ and the constraint set defined by Kraft inequality is closed and bounded (and hence compact), an optimal code length vector ${\bf l}^\dagger$ exists, and secondly the optimal code is given by (\ref{eq:solution1}). From the existence of the solution, it follows that for $\alpha \in [0,1)$,  $w_\alpha(x) > 0, \forall x \in {\cal U}$. This can also be deduced by noticing that the pay-off ${\mathbb L}_\alpha({\bf l}, {\bf p})$ is positive. As a result, all the weights $w_\alpha(x) > 0, \forall x \in {\cal U}$; otherwise, if there existed a negative weight $w_\alpha(x)$, one could have its corresponding codeword length to be large enough to make the pay-off ${\mathbb L}_\alpha({\bf l}, {\bf p})$ negative.

\noi From the characterization of optimal code length vector of Lemma~\ref{lemma_prelimin} and a well-known inequality, it follows that ${\mathbb L}_\alpha({\bf l}^\dagger, {\bf p})= -\sum_{x \in {\cal X}} w_\alpha(x) \log w_\alpha^\dagger(x)  \geq {\mathbb H}({\bf w}_\alpha)$ and equality holds if, and only if, $w_\alpha(x)=w_\alpha^\dagger(x), \forall x \in {\cal X}$. Therefore, for $\alpha \in [0,1)$ the weights satisfying (\ref{n3}) and corresponding to the optimal code length vector are uniquely represented via ${\bf w}_\alpha={\bf w}_\alpha^\dagger$. Moreover, by rounding off the optimal codeword lengths via $l^\ddagger(x) \tri \lceil -\log w_\alpha^\dagger(x)  \rceil$ Kraft inequality remains valid, while it is concluded that ${\mathbb H}({\bf w}_\alpha) \leq \sum_{x \in {\cal X}} l^\ddagger (x) w_\alpha(x) < {\mathbb H}({\bf w}_\alpha) +1$.

\noi The important observation concerning prefix code length vector ${\bf l}^\dagger \in {\mathbb R}_+^{|{\cal X}|}$ which minimizes the pay-off ${\mathbb L}_{\alpha}({\bf l}, {\bf p})=\sum_{x \in {\cal X}}w_\alpha(x) l(x) $ is that once the weight vector ${\bf w}_\alpha$ is identified for all $\alpha \in [0,1)$, then the optimal code is given by  ${ l}^\dagger(x)= -\log w_\alpha(x), ~\forall x \in {\cal X}$ and it is characterized for all $\alpha \in [0,1)$. The remaining part of this section is devoted to the problem of identifying the sets ${\cal U}$ and ${\cal U}^c$.


\noi The next lemma describes  monotonicity properties of the weight vector ${\bf w}_\alpha$ as a function of the probability vector ${\bf p}$, for all $\alpha \in [0,1)$.

\begin{lemma}
\label{lin}
Consider pay-off ${\mathbb L}_\alpha({\bf l}, {\bf p})$ and real-valued prefix codes. The following statements hold:
\begin{itemize}
 \item[1.] For $\{x,y\}  \subset  {\cal X}$, if $p(x)\leq p(y)$ then $w_\alpha(x) \leq w_\alpha(y)$, for all $\alpha \in [0,1)$. Equivalently,  $p(x_1) \geq p(x_2) \geq \ldots \geq p(x_{|{\cal X}|})>0$ implies $w_\alpha(x_1) \geq w_\alpha(x_2) \geq \ldots \geq w_\alpha(x_{|{\cal X}|})>0$, for all $\alpha \in [0,1)$.
 \item[2.] For $y \in {\cal U}^c$,  $w_\alpha(y)$ is a monotonically decreasing function of $\alpha \in [0,1)$.
 \item[3.]  For $x \in {\cal U}$, $w_\alpha(x)$ is a monotonically increasing  function of $\alpha \in [0,1)$.
 \end{itemize}
\end{lemma}

\begin{proof}
There exist three cases; more specifically,
\begin{enumerate}
\item $x,y \in \mathcal{U}^c$: then $ w_\alpha(x)=(1-\alpha)p(x) \leq (1-\alpha)p(y) = w_\alpha(y)$,~$\forall~\alpha \in [0,1)$;
\item $x,y \in \mathcal{U}$:  $w_\alpha(x)=w_\alpha(y) = w_\alpha^{*} \triangleq \min_{ x \in {\cal X}} w_\alpha(x)$;
\item $x \in \mathcal{U}$,  $y \in \mathcal{U}^c$ (or $x \in \mathcal{U}^c$,  $y \in \mathcal{U}$): consider the case $x \in \mathcal{U}$,  $y \in \mathcal{U}^c$. Since $x \in \mathcal{U}$ that means that $l(x)>l(y)$ and equivalently, $w_{\alpha}(y) > w_{\alpha}(x)$ Then, by taking derivatives we have
\begin{align}
\frac{\partial w_{\alpha}(y)}{\partial \alpha}&=-p(y)<0, \hst y \in {\cal U}^c,   \label{maxprobab} \\
\frac{\partial w_{\alpha}(x)}{\partial \alpha}&=\frac{\partial w_\alpha^{*}}{\partial \alpha}=\frac{1}{|\mathcal{U}|}\left(1-\sum_{ z \in \mathcal{U}}p(z)\right) >0, \hst x \in {\cal U}.  \label{minprobab}
\end{align}
According to \eqref{maxprobab} and \eqref{minprobab}, for $\alpha=0, w_\alpha(y)|_{\alpha =0}=p(y)\geq w_\alpha(x)|_{\alpha =0}=p(x)$, and as a function of $\alpha \in [0,1)$,   for $y \in \mathcal{U}^c$ the weight $w_{\alpha}(y)$ decreases, and for $x \in \mathcal{U}$ the weight $w_{\alpha}(x)$ increases. Hence, since $w_\alpha(\cdot)$ is a continuous function with respect to $\alpha$, at some $\alpha =\alpha^\prime$, $w_{\alpha^\prime}(x)=w_{\alpha^\prime}(y)=w^*_{\alpha^\prime}$.  Suppose that $w_\alpha(x)\neq w_\alpha(y)$, for some $\alpha=\alpha^\prime+d \alpha$, $d \alpha>0$.  Then, the largest weight will decrease and the smallest weight will increase as a function of $\alpha \in [0,1)$  according to \eqref{maxprobab} and \eqref{minprobab}, respectively. As a result, the two weight are moving together as a single weight.
\end{enumerate}
\end{proof}

\noi Before deriving the general coding algorithm, the following remark is introduced to illustrate  how the weights ${\bf w}_\alpha$ and the cardinality of the set ${\cal U}$ change as a function of  $\alpha \in [0,1)$.


\begin{remark}
\label{ex1}
Consider the special case when the probability vector ${\bf p}(x) \in {\mathbb P}({\cal X})$ consists of distinct probabilities, e.g.,   that $p(x_{|{\cal X}|})<p(x_{|{\cal X}|-1})$. The goal is to characterize the weights in a subset of $\alpha \in [0,1)$, such that  $w_\alpha(x_{|{\cal X}|})<w_\alpha(x_{|{\cal X}|-1})$ holds. Since ${\cal U}=\{x_{|{\cal X}|}\}$ and  $|{\cal U}|=1$, then
\begin{align*}
 {\mathbb L}_\alpha({\bf l}, {\bf p}) 
 &=\Big(\alpha+(1-\alpha) p(x_{|{\cal X}|}) \Big) l^*+ \sum_{x   \in    \mathcal{U}^c            }(1-\alpha) p(x) l(x) = \sum_{x \in {\cal X}} l(x) w_\alpha(x) ,
\end{align*}
\noi where  the weights are given by $w_{\alpha}(x) = (1-\alpha)p(x),~x \in \mathcal{U}^c$ and $w_{\alpha}(x_{|{\cal X}|}) = \alpha +(1-\alpha)p(x_{|{\cal X}|}) $ (by Lemma~\ref{lin}). For any $\alpha \in [0,1)$ such that the condition $w_\alpha(x_{|{\cal X}|}) < w_\alpha(x_{|{\cal X}|-1})$ holds, the optimal codeword lengths are given by $-\log w_\alpha(x), x \in {\cal X}$, and  this region of $\alpha \in [0,1)$ for which $|{\cal U}|=1$  is
\begin{align*}
\left\{\alpha \in [0,1):  \alpha +(1-\alpha)p(x_{|{\cal X}|})  < (1-\alpha)p(x_{|{\cal X}|-1}) \right\}.
\end{align*}
Equivalently,
\begin{align}\label{alpha1}
 \left\{ \alpha \in [0,1): \alpha < \frac{p(x_{|{\cal X}|-1})-p(x_{|{\cal X}|})}{1+p(x_{|{\cal X}|-1})-p(x_{|{\cal X}|})} \right\}.
\end{align}
Hence, under the condition $|{\cal U}|=1$ (i.e., $w_\alpha(x_{|{\cal X}|}) < w_\alpha(x_{|{\cal X}|-1})$), the optimal codeword lengths are given by $-\log w_\alpha(x), x \in {\cal X}$ for   $\alpha < \alpha_{1} \tri \frac{p(x_{|{\cal X}|-1})-p(x_{|{\cal X}|})}{1+p(x_{|{\cal X}|-1})-p(x_{|{\cal X}|})}$, while for     $\alpha \geq \alpha_{1}$ the form of the minimization problem changes, as more weights $w_{\alpha}(x)$ are such that $x \in \mathcal{U}$, and the cardinality of ${\cal U}$ is changed (that is, the partition of ${\cal X}$ into ${\cal U}$ and ${\cal U}^c$ is changed).  \\
Note that when $p(x_{|{\cal X}|}) = p(x_{|{\cal X}|-1})$, in view of the continuity of the weights ${\bf w}_\alpha$ as a function of $\alpha \in [0,1)$, the above optimal codeword lengths are only characterized for the singleton point $\alpha=\alpha_1={0}$, giving the classical codeword lengths. For $\alpha \in (\alpha_{1} ,1)$ the problem should be reformulated to characterize its solution over this region for which $|{\cal U}|\neq 1$. For example, if we consider the case for which $\alpha > \alpha_{1}$ and  $|{\cal U}|=2$ the problem can be written as
\begin{align*}
 {\mathbb L}_\alpha({\bf l}, {\bf p}) 
 &=\Big(\alpha+(1-\alpha) \big(p(x_{|{\cal X}|}+p(x_{|{\cal X}|-1})\big) \Big) l^*+ \sum_{x   \in    \mathcal{U}^c }(1-\alpha) p(x) l(x) = \sum_{x \in {\cal X}} l(x) w_\alpha(x) .
\end{align*}
 For any $\alpha \in [\alpha_{1},1)$ such that the condition $w_\alpha(x_{|{\cal X}|-1}) < w_\alpha(x_{|{\cal X}|-2})$ holds, the optimal codeword lengths are given by $-\log w_\alpha(x), x \in {\cal X}$ and this region is specified  by
\begin{align*}
\left\{\alpha \in [\alpha_{1},1): \frac{\alpha+(1-\alpha) \big(p(x_{|{\cal X}|}+p(x_{|{\cal X}|-1})\big)}{|{\cal U}|}  < (1-\alpha)p(x_{|{\cal X}|-2}) \right\}.
\end{align*}
Equivalently,
\begin{align}\label{alpha1}
 \left\{ \alpha \in [0,1): \alpha_{1} <\alpha < \frac{|{\cal U}| p(x_{|{\cal X}|-2})-(p(x_{|{\cal X}|})+p(x_{|{\cal X}|-1}))}{1+|{\cal U}| p(x_{|{\cal X}|-2})-(p(x_{|{\cal X}|})+p(x_{|{\cal X}|-1}))} \right\}.
\end{align}
\end{remark}

\ \

\noi Next, the merging rule which described how the weight vector ${\bf w}_\alpha$ changes as a function of  $\alpha \in [0,1)$ is identified, such that a solution to the coding problem is completely characterized  for arbitrary cardinality $|{\cal U}|$, and not necessarily distinct probabilities, for  any $\alpha \in [0,1)$. Clearly, there  is a  minimum $\alpha$ called $\alpha_{\max}$ such that  for any $\alpha \in [\alpha_{\max}, 1]$ there is no compression. This $\alpha_{\max}$ will be identified as well.

\ \

\noi Consider the complete characterization of the solution, as $\alpha$ ranges over $[0,1)$, for any initial probability vector ${\bf p}$ (not necessarily consisting of distinct entries).  Then, $|\mathcal{U}| \in \{1, 2, \ldots, |{\cal X}|-1\}$ while for $|{\cal U}|=|{\cal X}|$, $\alpha\in[\alpha_{\max},1]$, there is no compression since the weights are all equal. \\
 Define
\begin{align*}\label{ak}
\alpha_k & \tri \min\left\{\alpha \in [0,1): w_\alpha(x_{|{\cal X}|-(k-1)})= w_\alpha(x_{|{\cal X}|-k})\right\}, \hst    k \in \{1,\ldots, |{\cal X}|-1\},  \hst \alpha_0 \tri 0,  \\
\Delta \alpha_k &\tri \alpha_{k+1}-\alpha_k.
\end{align*}
By Lemma \ref{lin} the weights are ordered, hence  $\alpha_1$ is the smallest value of $\alpha \in [0,1)$ for which the smallest two weights are equal, $w_\alpha(x_{|{\cal X}|})= w_\alpha(x_{|{\cal X}|-1})$; $\alpha_2$ is the smallest value of $\alpha \in [0,1)$ for which the next smallest two weights are equal, $w_\alpha(x_{|{\cal X}|-1})= w_\alpha(x_{|{\cal X}|-2})$ and so forth, and $\alpha_{|{\cal X}|-1}$ is the smallest value of $\alpha \in [0,1)$ for which the two largest weights are equal, $w_\alpha(x_{2})= w_\alpha(x_{1})$. For a given value of $\alpha \in [0,1)$,  define the minimum over  $ x\in {\cal X}$ of the weights   by $w_{\alpha}^*  \tri \min_{ x \in {\cal X}} w_{\alpha}(x)$. \\
\noi Since for $k=0$, $w_{\alpha_0}(x)=w_{0}(x)=p(x), \forall x  \in {\cal X}$, is the set of initial symbol probabilities, let ${\cal U}_0$ denote the singleton set $\{x_{|{\cal X}|} \}$. Specifically,
\begin{align}
\mathcal{U}_0 \tri \left\{x\in \{x_{|{\cal X}|} \}:  p^{*}\tri \min_{ x\in {\cal X}} p(x)= p(x_{|{\cal X}|})  \right\} .
\end{align}
Similarly, ${\cal U}_1$ is defined as the set of symbols in $\{x_{|{\cal X}|-1}, x_{|{\cal X}|}\}$ whose weight evaluated at $\alpha_1$ is equal to the minimum weight  $w_{\alpha_1}^*$:
\bea
\mathcal{U}_1 \tri \Big\{ x \in \{ x_{|{\cal X}|-1}, x_{|{\cal X}|} \}:  w_{\alpha_1}(x)= w_{\alpha_1}^* \Big\}.
\eea
\noi In general, for  a given value of $\alpha_k, k \in \{1,\ldots, |{\cal X}|-1\}$, define
\begin{align}\label{Uk}
\mathcal{U}_k  \tri \Big\{ x \in \{ x_{|{\cal X}|-k},\ldots, x_{|{\cal X}|} \}:  w_{\alpha_k}(x)= w_{\alpha_k}^* \Big\}.
\end{align}

\begin{lemma}\label{prop1}
Consider  pay-off ${\mathbb L}_\alpha({\bf l}, {\bf p})$ and real-valued prefix codes. For $k \in \{0, 1, 2, \ldots, |{\cal X}|-1\}$ then
\begin{align}
w_\alpha(x_{ |{\cal X}|-k}) = w_{\alpha}(x_{|{\cal X}|})=w_{\alpha}^*, \hst \alpha \in [\alpha_k, \alpha_{k+1}) \subset [0,1).
\end{align}
Further, the cardinality of set $\mathcal{U}_k$ is $\left| \mathcal{U}_k \right|=k+1$, $k \in \{0, 1, 2, \ldots, |{\cal X}|-1\}$.
\end{lemma}
\begin{proof}
The validity of the statement  is shown by perfect induction.
\begin{align*}
 \mbox{Firstly, for} \hso \alpha=\alpha_{1}: \hst  w_\alpha(x_{|{\cal X}|})=w_\alpha(x_{|{\cal X}|-1}) \leq w_\alpha(x_{|{\cal X}|-2}) \leq \ldots \leq w_\alpha(x_{1}). \hst
\end{align*}
Suppose that, when $\alpha=\alpha_1+d\alpha \in [0,1)$, $d\alpha>0$, then $w_\alpha(x_{|{\cal X}|}) \neq w_\alpha(x_{|{\cal X}|-1}) $. Then,
\begin{align}
 {\mathbb L}_\alpha({\bf l}, {\bf p}) =\Big(\alpha+(1-\alpha)p(y) \Big) l^*+ \sum_{x \in \mathcal{U}_1^c }(1-\alpha) p(x) l(x), \nonumber
\end{align}
and the weights will be of the form $w_{\alpha}(x) = (1-\alpha)p(x)$ for $x \in \mathcal{U}_1^c $ and $w_{\alpha}(y) = \alpha +(1-\alpha)p(y)$ for $y\in {\cal U}_1 = \{x_{|{\cal X}|},x_{|{\cal X}|-1} \}$.
The rate of change of these weights with respect to $\alpha$ is
\begin{align}
\frac{\partial w_{\alpha}(x)}{\partial \alpha}&=-p(x)<0, ~x \in \mathcal{U}_1^c, \\
\frac{\partial  w_{\alpha}(y)}{\partial \alpha}&=1-p(y) >0, ~y \in \mathcal{U}_1 \label{minprob}.
\end{align}
Hence, the largest of the two would decrease, while the smallest would increase and therefore they meet again. This contradicts the  assumption that  $w_\alpha(x_{|{\cal X}|}) \neq w_\alpha(x_{|{\cal X}|-1}) $ for $\alpha>\alpha_1$, because otherwise one of the weights would be smaller and it should be increased with $\alpha$ as in \eqref{minprob}. Therefore, $w_\alpha(x_{|{\cal X}|}) = w_\alpha(x_{|{\cal X}|-1}), ~\forall \alpha \in [\alpha_1,1)$.\\
Secondly, for $\alpha>\alpha_k,, ~k \in \{2,\ldots, {|{\cal X}|}-1\}$,  suppose  the weights are
\begin{align*}
w_\alpha(x_{|{\cal X}|})=w_\alpha(x_{|{\cal X}|-1})= \ldots =w_\alpha(x_{|{\cal X}|-k})=w_{\alpha}^*.
\end{align*}
Then, the pay-off is written as
\begin{align*}
 {\mathbb L}_\alpha({\bf l}, {\bf p}) =\Big(\alpha+(1-\alpha)\sum_{x \in \mathcal{U}_k }p(x) \Big) l^*+ \sum_{x \in \mathcal{U}_k^c }(1-\alpha) p(x) l(x), \hst \alpha \in (\alpha_k, 1).
\end{align*}

\noi Hence,
\begin{align}
\frac{\partial w_{\alpha}(x)}{\partial \alpha}&=-p(x)<0, \hst x \in \mathcal{U}_k^c, \hso \alpha \in (\alpha_k, 1), \label{wi_prob}  \\
|\mathcal{U}_k |\frac{\partial w_{\alpha}^*}{\partial \alpha}&=1-\sum_{x \in \mathcal{U}_k }p(x)  >0,  \hst x \in \mathcal{U}_k, \hso \alpha \in (\alpha_k, 1).
\end{align}
Finally, in the case that $\alpha>\alpha_{k+1}, ~k \in \{2,\ldots, |{\cal X}|-2\}$, if any of the weights $w_\alpha( x),~x \in \mathcal{U}_k$,  changes differently than another, then, either at least one probability will become smaller than others and give a higher codeword length, or it will increase faster than the others and hence according to \eqref{wi_prob}, it will decrease to meet the other weights. 
Therefore, the change in this new set of probabilities should be the same, and the cardinality of $\mathcal{U}$ increases by one, that is, $|\mathcal{U}_{k+1}|=\left|k+2\right|, ~k\in \{ 2,\ldots |{\cal X}|-2 \}$.
\end{proof}

\ \

\noi Based on the results of Lemmas~\ref{lin} and~\ref{prop1}, the  next theorem describes how the weight vector ${\bf w}_\alpha$ changes as a function of $\alpha \in [0,1)$ so that the solution of the coding problem can be characterized.

\ \

\begin{theorem}\label{main_theorem}
Consider pay-off ${\mathbb L}_\alpha({\bf l}, {\bf p})$  and real-valued prefix codes. \\
For  $\alpha \in [\alpha_k,\alpha_{k+1})$, $k\in\{0,1,\ldots,|{\cal X}|-1\}$, the optimal weights
\begin{align*}
{\bf{w}^{\dagger}_{\alpha}} \tri \{{w}^{\dagger}_{\alpha}(x): x \in {\cal X}  \} \equiv \big({w}^{\dagger}_{\alpha}(x_1), {w}^{\dagger}_{\alpha}(x_2), \ldots, {w}^{\dagger}_{\alpha}(x_{|{\cal X}|})\big),
\end{align*}
 are given by
\begin{align}
{\small
w^{\dagger}_{\alpha}(x) =
\begin{cases}
(1-\alpha)p(x),~ x \in \mathcal{U}_{k}^{c}  \\ 
\displaystyle w_{\alpha_k}^* + (\alpha-\alpha_k) \frac{\sum_{x \in \mathcal{U}_{k}^{c}} p(x) }{|\mathcal{U}_{k}|}, ~x \in \mathcal{U}_{k},
\end{cases}
\label{weights_update}}
\end{align}
where $\mathcal{U}_k$ is given by \eqref{Uk} and
\begin{align}\label{akplus1}
\alpha_{k+1}=\alpha_k+(1-\alpha_k)\frac{(p(x_{|{\cal X}|-(k+1)})-p(x_{|{\cal X}|-k}) }{\frac{\sum_{x \in \mathcal{U}_{k}^{c}} p(x)}{|\mathcal{U}_{k}|}+p(x_{|{\cal X}|-(k+1)})} .
\end{align}
Moreover, the minimum $\alpha$, called $\alpha_{\max}$, such that for  $\alpha \in [\alpha_{\max},1] $  there is no compression, is given by
\begin{align}
\alpha_{\max}=1-\frac{1}{|\mathcal{X}|p(x_1)}.
\end{align}
\end{theorem}

{\begin{proof}
By Lemma \ref{prop1}, for $\alpha \in [\alpha_k, \alpha_{k+1})$, the lowest probabilities become equal and change together forming a total weight given by
\begin{align*}
\sum_{x \in \mathcal{U}_k }w_{\alpha}(x)&=|\mathcal{U}_{k}| w_{\alpha}^* = \alpha+(1-\alpha)p(x_{|{\cal X}|})+\ldots+(1-\alpha)p(x_{|{\cal X}|-k}).
\end{align*}
\noi Hence,
\begin{align}
|\mathcal{U}_{k}|\frac{ \partial  w_{\alpha}^*}{\partial \alpha}&=1-\sum_{x \in \mathcal{U}_k }p(x) , \\
\frac{ \partial w_\alpha^*}{\partial \alpha}&=\frac{1-\sum_{x \in \mathcal{U}_k }p(x) }{|\mathcal{U}_{k}|}={\frac{\sum_{x \in \mathcal{U}_{k}^{c}} p(x)}{|\mathcal{U}_{k}|}}.
\end{align}
By letting, $\delta_k(\alpha)\tri\alpha - \alpha_{k}$, then
\begin{align}\label{eq:sol1}
w_\alpha^* =w_{\alpha_k}^* +\delta_k(\alpha){\frac{\sum_{x \in \mathcal{U}_{k}^{c}} p(x)}{|\mathcal{U}_{k}|}}, \hst x \in {\mathcal{U}_{k}},
\end{align}
and $w_{\alpha}(x) =(1-\alpha)p(x), x \in {\mathcal{U}_{k}^{c}}$.
When $\delta_k(\alpha)|_{\alpha =\alpha_{k+1}}=\alpha_{k+1}-\alpha_k$,  then $w_{\alpha_{k+1}}(x_{|{\cal X}|-(k+1)})=w_{\alpha_{k+1}}^*$, and
\begin{align*}
&\left(1- \alpha_{k+1} \right) p(x_{|{\cal X}|-(k+1)}) = w_{\alpha_k}^* +\delta_k(\alpha_{k+1}){\frac{\sum_{x \in \mathcal{U}_{k}^{c}} p(x)}{|\mathcal{U}_{k}|}}. 
\end{align*}
After some manipulations,  $\alpha_{k+1}$ is given by
\begin{align}\label{alpha_k}
\alpha_{k+1}=\alpha_k+(1-\alpha_k)\frac{p(x_{|{\cal X}|-(k+1)})-p(x_{|{\cal X}|-k})}{\frac{\sum_{x \in \mathcal{U}_{k}^{c}} p(x)}{|\mathcal{U}_{k}|}+p(x_{|{\cal X}|-(k+1)})} .
\end{align}

\noi When there exist no compression all the weights are equal. Hence,
\begin{align}
w_{\alpha_{\max}}^*=\frac{\sum_{x\in \mathcal{X}}w_{\alpha_{\max}}(x)}{| \mathcal{X}|}=\frac{1}{| \mathcal{X}|}.
\end{align}
The minimum $\alpha$ beyond which there is no compression is the $\alpha$ at which all the weights become equal for the first time. This is the case when $(1-\alpha_{\max})p(x_1)=w_{\alpha_{\max}}^*$ or equivalently $
\alpha_{\max}=1-\frac{1}{|\mathcal{X}|p(x_1)}$.
\end{proof}

\ \

\noi Theorem~\ref{main_theorem} facilitates the computation of the optimal real-valued prefix codeword lengths vector ${\bf l}^\dagger$  minimizing pay-off ${\mathbb L}_{\alpha}({\bf l}, {\bf p})$ as a function of $\alpha \in [0,1)$ and the initial source probability vector ${\bf p}$, via re-normalization and merging. Specifically, the optimal weights are found recursively calculating $\alpha_k, k \in \{0,1,\ldots, |{\cal X}|-1\}$.  For any specific $\hat{\alpha} \in [0,1)$ an algorithm is given next, which describes  how to obtain the optimal real-valued prefix codeword lengths minimizing pay-off ${\mathbb L}_{\hat{\alpha}}({\bf l}, {\bf p})$.

\subsection{An Algorithm for Computing the Optimal Weights}\label{subsec:algorithm}
\label{fa}
\noi For any probability distribution  ${\bf p} \in {\mathbb P}({\cal X})$ and $\alpha \in [0,1)$ an algorithm is presented to compute the optimal weight vector ${\bf w}_\alpha$ of Theorem~\ref{main_theorem}. By Theorem~\ref{main_theorem} (see also Fig.~\ref{probs1} for a schematic representation of the weights for different values of $\alpha$), the weight vector ${\bf w}_\alpha$ changes piecewise linearly as a function of  $\alpha \in [0,1)$. The value of $\alpha_{\max}$ is also indicated.

\begin{figure}[H]
\centering
\includegraphics[width=\columnwidth]{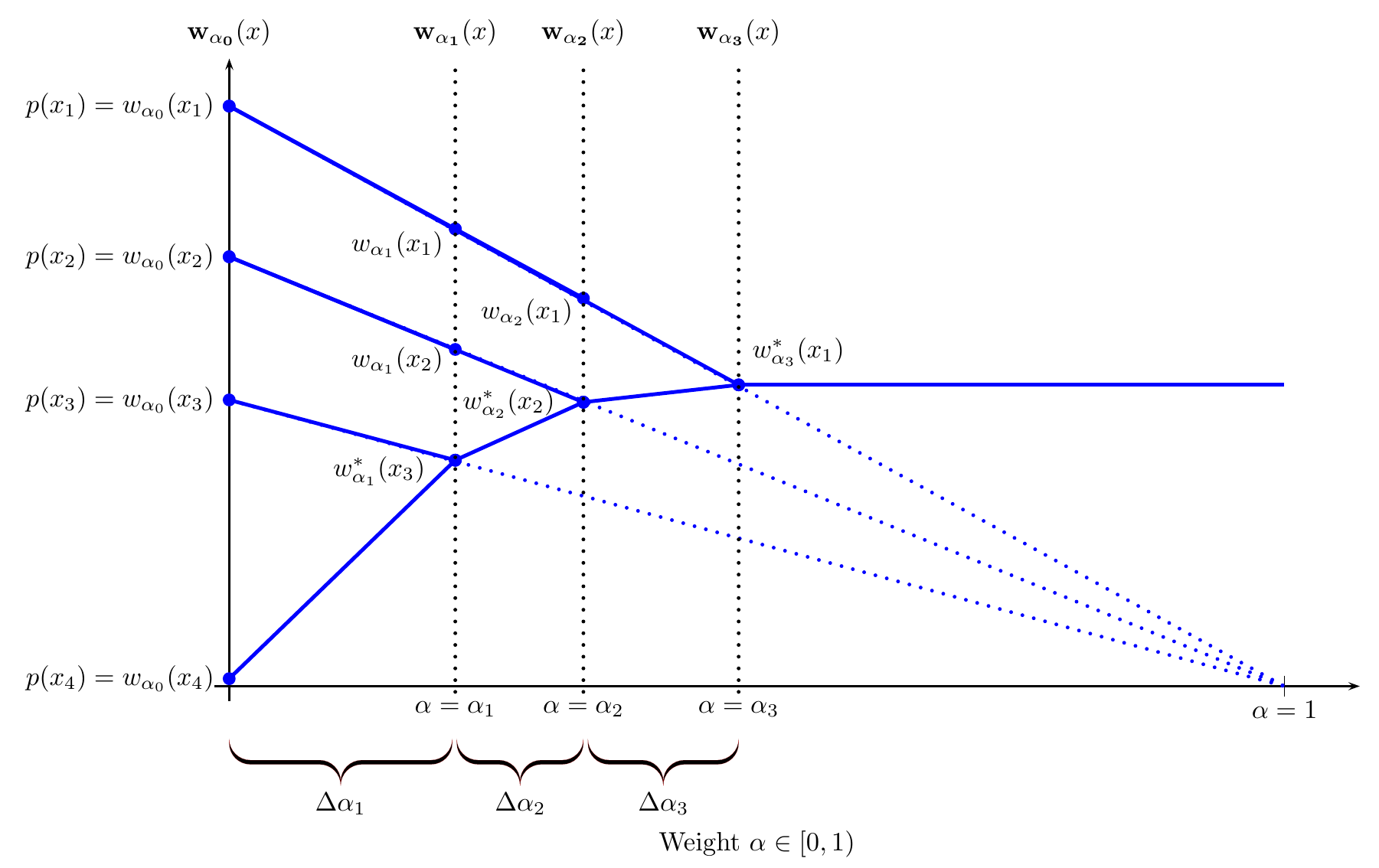}
\caption{A schematic representation of the weights for different values of $\alpha$.}\label{probs1}
\end{figure}

Given a  specific value of $\hat{\alpha} \in [0,1)$, in order to calculate the weights $w_{\hat{\alpha}}(x)$, it is sufficient  to determine  the values of $\alpha$ at the intersections by using \eqref{akplus1}, up to the value of $\alpha$ for which the intersection gives a value greater than $\hat{\alpha}$, or up to the last intersection (if all the intersections give a smaller value of $\alpha$) at $\alpha_{\max}$ beyond which there is no compression. For example, if $\alpha_1<\hat{\alpha}<\alpha_2$,  find all $\alpha$'s at the intersections up to and including $\alpha_2$ and subsequently, the  weights at $\hat{\alpha}$ can be found by using \eqref{weights_update}. Specifically,  check first if $\hat{\alpha}\geq \alpha_{\max}$. If yes, then the weights are equal to $1/|\mathcal{X}|$. If  $\hat{\alpha}< \alpha_{\max}$, then find $ \alpha_1,   \ldots, \alpha_m$, $m\in \mathbb{N}$, $m\geq 1$, until   ${\alpha_{m-1}}<\hat{\alpha}\leq {\alpha_{m}}$. As soon as the $\alpha$'s at the intersections are found, the  weights at $\hat{\alpha}$ can be found by using \eqref{weights_update}. The algorithm is easy to implement and extremely fast due to its low computational complexity. The worst case scenario appears when $\alpha_{|\mathcal{X}|-2}<\hat{\alpha}< \alpha_{\max}=\alpha_{|\mathcal{X}|-1}$, in which all $\alpha$'s at the intersections are required to be found. Note that, if $\alpha$ is closer to $\alpha_{\max}$, then it is easier  to find $\alpha_{\max}$ first and then to implement the algorithm backwards. In general, the worst case complexity of the algorithm is $\mathcal{O}(n)$. The complete algorithm is depicted under Algorithm \ref{charalambous_algorithm}.



\begin{algorithm}
\caption{\small Algorithm for Computing the Weight Vector ${\bf w}_\alpha$ for Problem~\ref{problem1}}
\label{charalambous_algorithm}
\begin{algorithmic}
\STATE $\,$\\
\STATE \textbf{initialize}
\STATE $\quad\, \mathbf{p}=\left(p(x_1), p(x_2), \ldots, p(x_{|{\cal X}|})\right)$, $\alpha = \hat{\alpha}$
\STATE $\quad\, k=0$,  $\alpha_0 = 0$,  $\alpha_{\max}=1-1/(|\mathcal{X}|p(x_1))$
\IF{$\hat{\alpha} \geq \alpha_{\max}$}
\RETURN $w^{\dagger}_{\hat{\alpha}}=1/|{\cal X}|$, $\forall x\in \mathcal{X}$
\ENDIF
\WHILE{$\displaystyle \alpha_{k}< \hat{\alpha}< \alpha_{\max}$}
\STATE {Calculate $\alpha_{k+1}$:}
\STATE {$\quad\, \displaystyle \alpha_{k+1}= \alpha_{k}+(1-\alpha_k)\frac{p(x_{|{\cal X}|-(k+1)})-p(x_{|{\cal X}|-k})}{\frac{\sum_{x \in \mathcal{U}_{k}^{c}} p(x)}{k+1}+p(x_{|{\cal X}|-(k+1)})}  $}
\STATE {$k \leftarrow k + 1$}
\ENDWHILE
\STATE {$k \leftarrow k - 1$}
\STATE {Calculate $\mathbf{w}^{\dagger}_{\hat{\alpha}}$:}
\FOR{$v = 1$ to $|{\cal X}|-(k+1)$}
\STATE $w^{\dagger}_{\hat{\alpha}}(x_{v})=(1-\hat{\alpha})p(x_{v})$
\STATE $v \leftarrow v + 1$
\ENDFOR
\STATE {Calculate $w^{*}_{\hat{\alpha}} $:}
\STATE $\quad\, \displaystyle w^{*}(\hat{\alpha}) =\left( 1- a_{k} \right)p(x_{|{\cal X}|-k})+ (\hat{\alpha}-\alpha_k) \frac{\sum_{x \in \mathcal{U}_{k}^{c}} p(x) }{k+1} $
\FOR{$v = |{\cal X}|-k$ to $|{\cal X}|$}
\STATE $\displaystyle w^{\dagger}(x_{v})=w^{*}_{\hat{\alpha}}$
\STATE $v \leftarrow v+ 1$
\ENDFOR
\RETURN $\mathbf{w}^{\dagger}_{\hat{\alpha}}$
\end{algorithmic}
\end{algorithm}

\section{Optimal Codeword Lengths}\label{sec:lengths}

\noi This section presents the complete characterization of the optimal real-valued codeword length vectors ${\bf l} \in {\cal L}\left({\mathbb R}_+^{|{\cal X}|}\right)$ of the pay-offs stated under Problem~\ref{problem1}. Further, a coding theorem is derived and relations to limited length coding and coding with general pay-off criteria are described. The related problems are stated under Problem~\ref{problem2}, Problem~\ref{problem3a}. Finally, the application of the new codes in the context of universal coding applications in which the source probability vector belongs to a specific class is discussed.  \\


\noi In view of Lemma~\ref{lemma_prelimin} (and the discussion following it) and Theorem~\ref{main_theorem} the main theorem which gives the optimal codeword length vector is presented.

 \begin{theorem}\label{mos}
Consider Problem~\ref{problem1} for any $\alpha \in [0,1)$. The optimal prefix  code ${\bf l}^\dagger \in {\mathbb {{R}}}_+^{|{\cal X}|}$ minimizing pay-off ${\mathbb L}_{\alpha}({\bf l}, {\bf p})$ is given by
\bea \label{eq:solutions}
l_{\alpha}^\dagger(x) = \left\{ \begin{array}{lll}
 -\log{\Big((1-\alpha)p(x)\Big)}=w_\alpha(x), &  x \in {\cal U}_k^c   \\
 -\log{\Big( \frac{\alpha+(1-\alpha)\sum_{x \in \mathcal{U}_k}p(x)}{  |\mathcal{U}_k|} \Big)}=w_\alpha(x), &  x \in {\cal U}_k . \end{array} \right.
\eea
Here  $\alpha \in [\alpha_k, \alpha_{k+1}) \subset [0,1) ,  k \in \{1, \ldots, |{\cal X}|-1\}$, and $\alpha_k, \alpha_{k+1}$ are found  from Theorem~\ref{main_theorem}.
 \end{theorem}

\begin{proof} (\ref{eq:solutions}) follows from Lemma~\ref{lemma_prelimin} while the specific $\alpha \in [\alpha_k, \alpha_{k+1})$ follow from Theorem~\ref{main_theorem}.
\end{proof}

\noi Note that for $\alpha =0$ Theorem~\ref{mos} corresponds to the Shannon solution $l^{sh}(x)=-\log p(x)$, while  for $\alpha \in [\alpha_{\max}, 1)$  the weight vector ${\bf w}_\alpha$ is identically distributed, and hence $l_\alpha^\dagger(x)= \frac{1}{|{\cal X}|}$. The behavior of $w_\alpha(x)$ and  $l_\alpha^\dagger(x)$ as a function of $\alpha \in [0,1)$ is described in the next Section via  illustrative examples.  Clearly, by rounding off the optimal codeword lengths via $l^\ddagger(x) \tri \lceil -\log \big(w_\alpha^\dagger(x) \big)  \rceil$ then ${\mathbb H}({\bf w}_\alpha) \leq \sum_{x \in {\cal X}} l^\ddagger (x) w_\alpha(x) < {\mathbb H}({\bf w}_\alpha) +1$.
Note that one may fix the minimum or maximum lengths in (\ref{eq:solutions}) and find the value of $\alpha \in [0,1)$ which gives these specific lengths. This observation will be discussed in detail in Section~\ref{limitedLength}.

\noi The following proposition shows that the optimal pay-off is non-decreasing and concave function of $\alpha$.

\begin{proposition}\label{prop_concavity}
The optimal pay-off ${\mathbb L}_{\alpha}({\bf l}^\dagger, {\bf p})$ is non-decreasing  concave function  of $\alpha \in [0,1)$.
\end{proposition}

\begin{proof}
See Appendix \ref{proof_concavity}.
\end{proof}

\subsection{Coding Theorem}\label{subcoding}

This section proves a coding theorem by considering sources which generate symbols independently. Let ${\cal X}^n\tri \times_{i=1}^n {\cal X}$ denote the $n$th extension of the source which generates symbols in ${\cal X}^n$  independently according to ${\bf p}\in {\mathbb S}({\cal X})$ (e.g., the extension source is memoryless). A typical realization of the $n$th extension source $x^n \in {\cal X}^n$ is an n-tuple of the form $x^n=(x_{i_1}, x_{i_2}, \ldots, x_{i_n}), x_{i_j} \in {\cal X}
, 1 \leq j \leq n$. Since the symbols are independently generated then $p(x^n)=p(x_{i_1})p(x_{i_2})\ldots p(x_{i_n})$. Let $l(x^n)$ denote the length of some uniquely decodable code for a given realization $x^n \in {\cal X}^n$. Then, the maximum and average length pay-off for such $n-$tuple sequences $x^n$ is defined by
\begin{align*}
{\mathbb L}_\alpha^n({\bf l},{\bf p}) & \tri \alpha \max_{x^n \in {\cal X}^n}l(x^n) + (1-\alpha) \sum_{ x^n \in {\cal X}^n} l(x^n) p(x^n) \nonumber \\
& = \Big( \alpha + (1-\alpha) \sum_{x^n \in {\cal U}^n} p(x^n)\Big)l^* + \sum_{x^n \in {\cal U}^{n,c}} (1-\alpha) p(x^n)l(x^n) \nonumber \\
&=\sum_{x^n \in {\cal X}^n} w_\alpha(x^n) l(x^n), \hst \alpha \in [0,1),
\end{align*}
where $l^* \tri \max_{x^n \in {\cal X}^n}l(x^n), {\cal U}^n \tri \Big\{x^n \in {\cal X}^n: l(x^n)= l^*\Big\}, {\cal X}^n= {\cal U}^n \cup {\cal U}^{n,c}$, and $\sum_{x^n \in {\cal U}^n} w_\alpha(x^n)= \alpha + (1-\alpha) \sum_{x^n \in {\cal U}^n} p(x^n)$, $w_\alpha(x^n)= (1-\alpha)p(x^n), x^n \in {\cal U}^{n,c}$.
Let $l(x^n)$ be the integer length vector which satisfies
\bea
  -  \log w_\alpha(x^n)  \leq l(x^n) < -\log w_\alpha(x^n) +1 \label{it1}
  \eea
where
\bea
w_{\alpha}(x^n) = \left\{ \begin{array}{lll}
 (1-\alpha)p(x^n), &  x^n \in {\cal U}^{n,c}   \\
  \frac{\alpha+(1-\alpha)\sum_{x^n \in \mathcal{U}^n}p(x^n)}{  |\mathcal{U}^n|} , &  x^n \in {\cal U}^n . \end{array} \right.
\eea
Then the  maximum and average length pay-off per source symbol  $\frac{1}{n}{\mathbb L}_\alpha^n({\bf l},{\bf p})$ satisfies
\bea
\frac{1}{n} {\mathbb H}(w_\alpha(x^n)) \leq  \frac{1}{n}{\mathbb L}_\alpha^n({\bf l},{\bf p}) < \frac{1}{n} {\mathbb H}(w_\alpha(x^n))+\frac{1}{n} . \label{ct2}
\eea
Hence, by choosing $n$ sufficiently large then $\frac{1}{n}{\mathbb L}_\alpha^n({\bf l},{\bf p})$ can be made arbitrarily close to the lower bound $\frac{1}{n} {\mathbb H}(w_\alpha(x^n))$. Define the entropy rate of $w_\alpha(x^n)$ by
\bea
{\cal H}({\bf w}_\alpha) \tri \lim_{n \rar \infty} \frac{1}{n} {\mathbb H}(w_\alpha(x^n)) \label{ct3}
\eea
Then, the following coding theorem is obtained.

 \begin{theorem}
 \label{cthm1}
Consider a discrete source with alphabet ${\cal X}$ generating symbols independently according to ${\bf p} \in {\mathbb S}({\cal X})$. Then by encoding uniquely decodable sufficiently long sequences of $n$ source symbols it is possible to make the maximum and average length pay-off per source symbol  $\frac{1}{n}{\mathbb L}_\alpha^n({\bf l},{\bf p})$ arbitrarily  close the entropy rate ${\cal H}({\bf w}_\alpha)$. Moreover, it is not possible to find a uniquely decodable code whose maximum and average length pay-off per source symbol  $\frac{1}{n}{\mathbb L}_\alpha^n({\bf l},{\bf p})$ is less than the entropy rate ${\cal H}({\bf w}_\alpha)$.
 \end{theorem}

\begin{proof} The first part of the theorem follows by the above discussion. The second part of the theorem follows from the discussion below Lemma~\ref{lemma_prelimin}.
\end{proof}

%
%
%
%
\subsection{Limited-Length Shannon Coding}\label{limitedLength}

\noi Note that from the characterization of optimal codes for Problem~\ref{problem1},  one can also obtain as a special case  the characterization of  optimal codes minimizing the average codeword length subject to a hard constraint on the maximum codeword length, as defined below.

\begin{problem}\label{problem3a}
Given a known source probability vector ${\bf p} \in   {\mathbb P}({\cal X})$ and  a hard constraint $L_{\lim} \in [1,\infty)$, find a prefix code length vector ${\bf l}^* \in {\mathbb R}_+^{|{\cal X}|}$ which minimizes the Average Length Subject to Maximum Length Constraint pay-off ${\mathbb L}({\bf l}, {\bf p})$ defined by
\begin{subequations}\label{eq:problem3a}
\begin{align}
& \mathbb{ L}({\bf l}, {\bf p})  \tri \sum_{x\in {\cal X}} l(x) p(x),  \label{j3a}\\
& \text{subject to} \quad  \max_{x\in {\cal X}} l(x) \leq L_{\lim}. \label{j3b}
\end{align}
\end{subequations}
\end{problem}

\noi Limited length coding problems are of interest in various applications, such as distributed systems that are delay-sensitive and require short codewords or/and fast coders with short code table size.\\
\noi It is important to note that the solution of Problem~\ref{problem3a} does not in general give the solution of Problem~\ref{problem1}. For inter-valued prefix codes ${\bf l}^* \in {\mathbb Z}_+^{|{\cal X}|}$, the solution of Problem~\ref{problem3a} is addressed in \cite{2010:golinZhang} via a dynamic programming approach. This led to the so-called length-limited Huffman algorithm investigated extensively in the literature (for more details, see \cite{2010:golinZhang} and references therein).

\noi Here it is noticed that by introducing a real-valued Lagrange multiplier $\mu$ associated with the constraint on the maximum length the unconstrained pay-off is defined by
\begin{align} \label{Lag_1}
 {\mathbb L}({\bf l}, {\bf p}, \mu) & \tri \sum_{x \in {\cal X}} l(x) p(x) +  \mu(\max_{x\in {\cal X}} l(x)- L_{\lim}), \quad \mu>0 \nonumber \\
 & =\mu \max_{x\in {\cal X}} l(x) + \sum_{x \in {\cal X}} l(x) p(x) -\mu L_{\lim}.
\end{align}
Hence, the optimal code from Problem~\ref{problem3a} is obtained from the optimal code solution of  Problem 1, by substituting $\mu=\alpha/(1-\alpha)$, and then relating the value of the  Lagrange multiplier with a specific value of $\alpha$ for which the codeword lengths will be limited by $L_{\lim}$. The complete characterization of the optimal codes and the associated coding algorithm are given next.

\begin{theorem}\label{th:limited1}
Consider Problem~\ref{problem3a} for any $\alpha \in [0,1)$. The optimal prefix code ${\bf l}^\dagger \in {\mathbb R}_+^{|{\cal X}|}$ minimizing the pay-off ${\mathbb L}({\bf l}, {\bf p})$ is given by
\bes
l_{\alpha}^\dagger(x) = \left\{ \begin{array}{lll}
 -\log{\Big((1-\alpha)p(x)\Big)}, & x \in {\cal U}^{c}_k  \\
 -\log{\Big( \frac{\alpha+(1-\alpha)\sum_{x \in \mathcal{U}_k}p(x)}{  |\mathcal{U}_k|} \Big)}, & x \in {\cal U}_k \end{array} \right.
\ees
where
\bes
\alpha=
\begin{cases} 0 & \text{, if $-\log (p_{|\mathcal{X}|})\leq L_{\lim}$,}
\\
1-\frac{1-|\mathcal{U}_k|D^{-L_{\lim}}}{\sum_{x \in \mathcal{U}_{k}^{c}}p(x)} & \text{, if $-\log \left(\frac{\sum_{x\in \mathcal{X}}p(x)}{|\mathcal{X}|} \right) < L_{\lim} \leq -\log (p(x_{|\mathcal{X}|}))$, }
\\
\alpha_{\max} & \text{, if $ L_{\lim} = -\log \left(\frac{\sum_{x\in \mathcal{X}}p(x)}{|\mathcal{X}|} \right)$. }
\end{cases}
\ees
If $L_{\lim}<-\log\left(\frac{\sum_{x\in \mathcal{X}}p(x)}{|\mathcal{X}|} \right)$, there is no feasible solution to Problem \ref{problem3a}.
\end{theorem}

\begin{proof}
\noi Note that pay-off ${\mathbb L}({\bf l}, {\bf p})$ is a convex function and the constraint set is convex, hence  this is  a convex optimization problem. By introducing a real-valued Lagrange multiplier $\mu$ associated with the maximum length constraint, the augmented pay-off is equivalent to the pay-off \eqref{b30} in Problem \ref{problem1}, by setting $\mu \tri \alpha /(1-\alpha)$. The bound on the maximum length, $L_{\lim}$, determines the value of $\alpha$ for which the maximum codeword length is less than $L_{\lim}$. If $L_{\lim}$ is greater than or equal to the maximum codeword length for $\alpha=0$ (i.e., for $\alpha=0$, $\| l \|_{\infty}=-\log (p_{|\mathcal{X}|})$), then the maximum codeword length for all $\alpha \in [0,1)$ will be smaller than $L_{\lim}$. If $L_{\lim}$ is smaller than the length $-\log\left(\frac{\sum_{x\in \mathcal{X}}p(x)}{|\mathcal{X}|} \right)$, for which there is no compression, then the maximum length cannot be smaller, and therefore, for $\alpha> \alpha_{\max}$ there is no feasible solution to Problem \ref{problem3a}. If, however, $-\log \left(\frac{\sum_{x\in \mathcal{X}}p(x)}{|\mathcal{X}|} \right) \leq L_{\lim} \leq -\log (p(x_{|\mathcal{X}|}))$, then using the expression for the maximum length in Theorem \ref{mos} we have
\begin{align*}
L_{\lim}= -\log{\Big( \frac{\alpha+(1-\alpha)\sum_{x \in \mathcal{U}_k}p(x)}{  |\mathcal{U}_k|} \Big)}.
 \end{align*}
 As a result,  after manipulation $\alpha$  is given by
 \begin{align}
 \alpha=1-\frac{1-|\mathcal{U}_k|D^{-L_{\lim}}}{\sum_{x \in \mathcal{U}_{k}^{c}}p(x)}. \label{aLmax}
\end{align}
\noi From  \eqref{aLmax}, it is evident that the cardinality of ${\cal U}_k$,  $|\mathcal{U}_k|$ and the symbols $x\in\mathcal{U}_k$ should be known, in order to calculate ${\sum_{x \in \mathcal{U}_{k}^{c}}p(x)}$.
\end{proof}

\ \

\noi Next, a new algorithm (Algorithm \ref{charalambous_limitedLength}) is introduced to calculate the complete solution of Problem \ref{problem3a}, and hence the value of $\alpha$ and the weight vector ${\bf w}_\alpha$ such that the maximum codeword length is upper bounded by $L_{\lim}$.
\begin{algorithm}
\caption{\small Algorithm for Computing $\alpha$ and the Weight Vector ${\bf w}_\alpha$ for Problem~\ref{problem3a}}
\label{charalambous_limitedLength}
\begin{algorithmic}
\STATE $\,$\\
\STATE \textbf{initialize}
\STATE $\quad\, \mathbf{p}=\left(p(x_1), p(x_2), \ldots, p(x_{|{\cal X}|})\right)$, $k=0$, $\alpha_0 = 0$,
\STATE $\quad\, l_{\max}=-\log(\min_{x\in \mathcal{X}}p(x))$, $l_{\min}=-\log \left(\frac{\sum_{x\in \mathcal{X}}p(x)}{|\mathcal{X}|} \right)$
\IF{$L_{\lim}>l_{\max}$}
\STATE {$\alpha=0$}
\ELSIF{$L_{\lim}<l_{\min}$}
\STATE {No feasible $\alpha$ exist.}
\ELSE
%
\WHILE{$L_{\lim}<l_{\max} $}
\STATE {$ \displaystyle \alpha_{k+1}= \alpha_{k}+(1-\alpha_k)\frac{p(x_{|{\cal X}|-(k+1)})-p(x_{|{\cal X}|-k})}{\frac{\sum_{x \in \mathcal{U}_{k}^{c}} p(x)}{k+1}+p(x_{|{\cal X}|-(k+1)})}  $}
\STATE $\displaystyle w^{*}_{{\alpha_{k+1}}} =\left(1- a_{k+1} \right)p(x_{|{\cal X}|-(k+1)})$
\STATE $ \displaystyle l_{\max}= -\log(w^{*}_{{\alpha_{k+1}}} )$
\STATE {$k \leftarrow k + 1$}
\ENDWHILE
\STATE {$k \leftarrow k - 1$}
\STATE $\displaystyle \hat{\alpha}=1-\frac{1-|k+1|D^{-L_{\lim}}}{\sum_{x\in \mathcal{U}_k^c}p(x)}$
\STATE $\,$\\
\FOR{$v = 1$ to $|{\cal X}|-(k+1)$}
\STATE $w^{\dagger}_{\hat{\alpha}}(x_{v})=(1-\hat{\alpha})p(x_{v})$
\STATE $v \leftarrow v + 1$
\ENDFOR
\STATE $\displaystyle w^{*}_{\hat{\alpha}} =\left( 1- a_{k} \right)p(x_{|{\cal X}|-k})+ (\hat{\alpha}-\alpha_k) \frac{\sum_{x \in \mathcal{U}_{k}^{c}} p(x) }{k+1} $
\FOR{$v = |{\cal X}|-k$ to $|{\cal X}|$}
\STATE $\displaystyle w^{\dagger}(x_{v})=w^{*}_{\hat{\alpha}}$
\STATE $v \leftarrow v+ 1$
\ENDFOR

\ENDIF

\RETURN $\mathbf{w}^{\dagger}_{\hat{\alpha}}$
\end{algorithmic}
\end{algorithm}

Even though the two algorithms \ref{charalambous_algorithm} and \ref{charalambous_limitedLength} are similar, there exist some basic differences. Algorithm \ref{charalambous_algorithm} has a certain value of $\alpha$ for which it tries to identify the cardinality of $\mathcal{U}$ and hence, specify the weight vector $\mathbf{w}_\alpha$. On the other hand, algorithm \ref{charalambous_limitedLength} uses the maximum length to find if there exist a feasible $\alpha$ for which the limited-length constraint is fulfilled. Then, if feasibility is guaranteed, the cardinality is specified by comparing the optimum lengths at the merging points with the specified maximum length. Therefore, given the cardinality, the corresponding $\alpha$ is specified and finally, in the same way as in algorithm \ref{charalambous_algorithm}, the weight vector $\mathbf{w}_\alpha$ is specified.

%
%
%
%
\subsection{General Pay-Off and Limiting Problem}\label{limit}

\noi Problem~\ref{problem1} can be further  modified by noticing that  $\frac{1}{t}\log \sum_{ x \in {\cal X}} p(x) D^{t l(x)}$ is a nondecreasing function of $t \in [0,\infty)$, and $\lim_{t \rar \infty} \frac{1}{t}\log \sum_{ x \in {\cal X}} p(x) D^{t l(x)}=\max_{ x \in {\cal X}} l(x)$. Hence, by replacing  $\alpha \max_{ x \in {\cal X}} l(x)$ in ${\mathbb L}_{\alpha}({\bf l}, {\bf p})$,  by the function $\frac{\alpha}{t} \log \Big(\sum_{ x \in {\cal X}} p(x) D^{t \l(x)}\Big)$, the resulting pay-off takes into account moderate values below  $\max_{ x \in {\cal X}} l(x)$, obtaining a two-parameter pay-off. The pay-off resulting from this observation is defined next, while the solution is discussed.
\begin{problem}\label{problem2}
Given a known source probability vector ${\bf p} \in   {\mathbb P}({\cal X})$, weighting parameter $\alpha \in [0,1) $, and parameter $t \in (-\infty, \infty)$,   find a prefix code length vector ${\bf l}^* \in {\mathbb R}_+^{|{\cal X}|}$ which minimizes the two-parameter  Average of Linear and Exponential Functions of Length  pay-off ${\mathbb L}_{t,\alpha}({\bf l}, {\bf p})$   defined by
\begin{align}
{\mathbb L}_{t,\alpha}({\bf l}, {\bf p}) &  \tri \frac{\alpha}{t} \log \Big(\sum_{ x \in {\cal X}} p(x) D^{t \l(x)}\Big) + (1-\alpha) \sum_{ x \in {\cal X}} l(x) p(x), \label{b30aa}
\end{align}
for all $ \alpha \in [0,1)$ and $t \in (-\infty,\infty)$.

\end{problem}

\noi Although, the solution of Problem~\ref{problem2} will be investigated for $t \in [0, \infty)$, the problem is also well defined for $t \in (-\infty,0)$. The above pay-off is  a convex combination of the average of an  exponential function of the codeword length, and the average codeword length. However, moderate values of $t \in [0,\infty)$ are also of interest since the pay-off ${\mathbb L}_{t,\alpha}({\bf l}, {\bf p})$ can be interpreted as a trade-off between universal codes and average length codes.
Thus, for a fixed value of $\alpha \in [0,1)$, and since ${\mathbb L}_{t, \alpha}({\bf l}, {\bf p})$ is non-decreasing with respect to parameter  $t \in (0,\infty)$, then $t$  is another design parameter, which can be selected so that the average codeword length is below ${\mathbb L}_\alpha({\bf l}, {\bf p})$.  \\

\noi  The case $\alpha=1$ is investigated in \cite{1965:campbell_coding,1981:humblet_generalization,2006a:Baer,2008a:Baer, 2008b:Baer,1991:Merhav}, where relations to minimizing  buffer overflow probability are discussed. Further, it is not difficult to verify that ${\mathbb L}_{t,\alpha}({\bf l}, {\bf p})|_{\alpha=1}$ is also the dual problem of universal coding problems, formulated as a minimax, in which the maximization is over  a class of probability distributions which satisfy a relative entropy constraint with respect to a given fixed nominal probability distribution  \cite{2005:rezaei_bambos,2009:Gawrychowski_Gagie}. Hence, the pay-off ${\mathbb L}_{t,\alpha}|({\bf l}, {\bf p})_{\alpha=1}$ encompasses a trade-off between universal codes and buffer overflow probability  and average codeword length codes. \\
Similarly as in Problem~\ref{problem1}, a slight modification of the two-parameter pay-off to the  convex combination of the average of an exponential function of the pointwise redundancy and  the average pointwise redundancy, ${\mathbb L}_{t,\alpha}({\bf l}+\log {\bf p}, {\bf p})$, is of interest for integer-valued codes, since the real-valued codes minimizing this pay-off are $l^*(x)=-\log p(x), x \in {\cal X}$. To the best of our knowledge only the special cases of $\alpha=0, \alpha =1$ are investigated for pay-off ${\mathbb L}_{t,\alpha}({\bf l}+\log {\bf p}, {\bf p})$  (see \cite{2004:DrmotaSzpankowski,2006a:Baer,2008a:Baer, 2008b:Baer}).

 \begin{theorem}
 \label{mos1}
 Consider Problem~\ref{problem2} for any $\alpha \in [0,1)$, $t \in [0, \infty)$. The optimal prefix code ${\bf l}^\dagger \in {\mathbb R}_+^{|{\cal X}|}$ minimizing the pay-off ${\mathbb L}_{t,\alpha}({\bf l}, {\bf p})$ is given by
\bea
l_{t,\alpha}^\dagger(x) = -\log \Big( \alpha \nu_{t,\alpha}(x) + (1-\alpha) p(x)\Big), \hst \ x \in {\cal X}, \label{gsl}
\eea
where $\{\nu_{t,\alpha}(x): x \in {\cal X}\}$ is defined via the tilted probability distribution
\bea
\nu_{t,\alpha}(x) \tri \frac{D^{ t \: l_{t,\alpha}^\dagger(x)} p(x)}{ \sum_{ x \in {\cal X}} p(x) D^{ t \: l_{t,\alpha}^\dagger(x)} }, \hst   x \in {\cal X}.  \label{gs11}
\eea
 \end{theorem}

\begin{proof} See Appendix \ref{proof_mos1}.
\end{proof}

\noi Note that the solution stated under Theorem~\ref{mos1} corresponds, for $\alpha=0$ to the Shannon code, which minimizes the average codeword length pay-off, while for $\alpha=1$ (after manipulations) it is given by
\begin{align}
 l_{t,\alpha=1}^{\dagger}(x) = - \frac{1}{1+t} \log p(x)  + \log \Big( \sum_{x \in {\cal X}} p(x)^\frac{1}{1+t}   \Big), \hso x \in {\cal X}. \label{gs1}
\end{align}
Thus, (\ref{gs1}) is precisely the solution of a variant of the Shannon code,  minimizing the average of an exponential function of the codeword length pay-off \cite{1981:humblet_generalization,2008b:Baer}. It can be shown that
\bea
{\mathbb H}_{\frac{1}{1+t}}({\bf p}) \leq \sum_{x \in {\cal X}} p(x) l_{t,\alpha=1}^{\dagger}(x) < {\mathbb H}_{\frac{1}{1+t}}({\bf p}) +1 \label{b5}
 \eea
where $H_{a}({\bf p})$ is the  R\'enyi entropy given by
 \bea
    {\mathbb H}_{a}({\bf p}) \tri \frac{1}{1-a} \log \Big( \sum_{x \in {\cal X}} p(x)^a \Big), \hst a \tri \frac{1}{1+t}, \hso t \neq 0.  \label{b3}
 \eea
However, for any $\alpha \in (0,1)$ the following  system of equations should be solved.
\bea
D^{-l^\dagger(x)}=\alpha \frac{D^{ t \: l^\dagger(x)} p(x)}{ \sum_{ x \in {\cal X}} p(x) D^{ t \: l^\dagger(x)} } + (1-\alpha )p(x), \hst \forall  x \in {\cal X}. \label{gs11}
\eea

Although, the solution of Problem~\ref{problem1} is different from the solution of Problem~\ref{problem2}, in the limit, as $t \rar \infty$, the solutions should coincide,  provided the merging rule on how the solution changes with $\alpha \in [0,1)$ is employed. To this end, consider the following identities.
\bea
 \lim_{t \rar \infty} \frac{1}{t}\log \Big(\sum_{ x \in {\cal X}} p(x) D^{t l(x)}\Big)=\max_{ x \in {\cal X}} l(x), \hst  {\mathbb L}_{\infty,\alpha}({\bf l}, {\bf p}) \tri  \lim_{t \rar \infty} {\mathbb L}_{t,\alpha}({\bf l}, {\bf p})= {\mathbb L}_{\alpha}({\bf l}, {\bf p}). \label{l}
\eea
 Since the pay-off ${\mathbb L}_{t,\alpha}({\bf l}, {\bf p})$ is  in the limit, as $t \rar \infty$, equivalent to $\lim_{t \rar \infty} {\mathbb L}_{t,\alpha}({\bf l}, {\bf p})  ={\mathbb L}_{\alpha}({\bf l}, {\bf p})$, $\forall \alpha \in [0,1)$, then the codeword length vector minimizing ${\mathbb L}_{t,\alpha}({\bf l}, {\bf p})$ is expected to converge in the limit as $t \rar \infty$, to that which minimizes ${\mathbb L}_{\alpha}({\bf l}, {\bf p})$. To verify this claim consider the behavior of the optimal two parameter pay-off ${\mathbb L}_{\alpha}({\bf l}_{t,\alpha}^\dagger, {\bf p})$, for a fixed $\alpha \in [0,1)$ as $t$ increases, given in Theorem~\ref{mos1},  which is equivalent to
\bea
D^{-l_{t,\alpha}^\dagger(x)}=\alpha \frac{D^{ t \: l_{t, \alpha}^\dagger(x)} p(x)}{ \sum_{ x \in {\cal X}} p(x) D^{ t \: l_{t,\alpha}^\dagger(x)} } + (1-\alpha )p(x), \hst \forall  x \in {\cal X}. \label{gs11a}
\eea
Write
\begin{align*}
\sum_{ x \in {\cal X}}p(x) D^{t \: l_{t,\alpha}^\dagger(x)}= \sum_{ x \in {\cal U}^{c}} p(x) D^{t \: l_{t,\alpha}^\dagger(x)} + \sum_{ x \in {\cal U}} p(x) D^{t \: l_{t,\alpha}^\dagger(x)}.
\end{align*}
Utilizing the validity of the  limits under (\ref{l}), in the limit as, $t \rar \infty$, (\ref{gs11a}) becomes
\begin{align}
D^{-l_\alpha^\dagger(x)} & = (1-\alpha )p(x), \hst   x \in {\cal U}_k^c \label{gss1} \\
D^{-l_\alpha^\dagger(x)} &=\alpha \frac{ p(x)}{ \sum_{ x \in {\cal U}_k} p(x) } + (1-\alpha )p(x), \hst   x \in {\cal U}_k .\label{gss2}
\end{align}
Since $p(x)=p(y), \forall x, y \in {\cal U}_k$, then (\ref{gss1}) and (\ref{gss2}) are the same as  \eqref{weights_update}.
These calculations verify that $\lim_{t \rar \infty} {\mathbb L}_{t,\alpha}({\bf l}, {\bf p})={\mathbb L}_{\alpha}({\bf l}, {\bf p}), ~\forall {\bf l},  \mbox{and at} \hso {\bf l} ={\bf l}^\dagger$.
The point to be made here is that the solution of Problem~\ref{problem1} can be deduced from the solution of Problem~\ref{problem2}, in the limit as $t \rar \infty$, provided the merging rule on how the solution changes with $\alpha \in [0,1)$ is employed.

%
%
%
%

\subsection{Generalizations: Connections to Universal Coding}
\label{sec:universal}
\noi Although, the current paper does not investigate universal coding problems, this exposition is  included for the purpose of demonstrating that the optimal codes characterized  under Problem~\ref{problem1}, can be used to address problems of universal coding, having pay-off ${\mathbb L}_{t,\alpha}({\bf l}, {\bf p})$ or ${\mathbb L}_{\alpha}({\bf l}+\log {\bf p}, {\bf p})$, and  probability vector ${\bf p}$  belonging to a class of source probability vectors.  \\
Recall that universal coding and universal modeling \cite{1996:Rissanen}, and the so-called Minimum Description Length (MDL) principle  and Stochastic Complexity \cite{1998:barron} are often examined when the source probability distribution ${\bf p}$  is unknown,  modeled via a parameterized class ${\bf p}_{\theta} \tri \Big\{p_\theta(x): x \in {\cal X}, \theta \in {\Theta} \subset \Re^d\Big\}$ ($\theta$ is a  parameter vector), or a non-parameterized class ${\mathbb  S}({\cal X}) \subset {\mathbb P}({\cal X})$. Universal coding initiated in \cite{1973:davisson,1980:davisson_Leon-Garcia}, and further investigated in \cite{1981:KrichevskyTrofimov,1987:Shtarkov} aims at constructing a code for sequences of symbols generated by unknown sources, ${\bf p}_{\theta}$ or   ${\mathbb  S}({\cal X}) $,   such that as the length of the sequence increases, the average code length converges to the entropy of the true source that generated the sequence.\\
\noi When the source probability vector is not a singleton set, but a family or a class of probability vectors, then Problem~\ref{problem1} can be re-formulated to account for this generality as follows.

\begin{problem}\label{problem4}
Given a family of  source probability vectors ${\bf p} \in   {\mathbb S}({\cal X}) \subset {\mathbb P}({\cal X})$ and weighting parameter $\alpha \in [0,1) $, define the one parameter pay-offs as follows. \\
{\bf A.} Worst Case Maximum and Average Length.
\bea
{\mathbb L}_{\alpha}^+({\bf l}, {\bf p}) \tri  \max_{  {\bf p} \in   {\mathbb S}({\cal X})}  \Big\{ \alpha \max_{ x\in {\cal X}} l(x) + (1-\alpha) \sum_{ x \in {\cal X}} l(x) p(x)\Big\}. \label{b300}
\eea
{\bf B.} Worst Case Maximum and Average Redundancy.
 \begin{align}
  {\mathbb L}_{\alpha}^+({\bf l}+\log {\bf p}, {\bf p}) & \tri    \max_{  {\bf p} \in   {\mathbb S}({\cal X})}  \Big\{
 \alpha \max_{ x \in {\cal X}} \Big( l(x) + \log p(x)\Big)  + (1-\alpha) \Big( \sum_{ x \in {\cal X}} l(x) p(x) - {\mathbb H}({\bf p})\Big) \Big\}. \label{lm300}
 \end{align}
 The objectives are the following.
 \bi
\item  Find a prefix code length vector ${\bf l}^* \in {\mathbb R}_+^{|{\cal X}|}$ which minimizes the pay-off ${\mathbb L}_{\alpha}^+({\bf l}, {\bf p})$,

 \item  Find a prefix code length vector ${\bf l}^* \in {\mathbb R}_+^{|{\cal X}|}$ which minimizes the pay-off ${\mathbb L} {\mathbb R}_{\alpha}^+({\bf l}+\log {\bf p}, {\bf p})$,

\ei
for all $ \alpha \in [0,1)$.
\end{problem}

\noi The universal coding problems defined above are based on minimax techniques, the minimization being over the codeword lengths satisfying Kraft inequality, the maximization being over the class of probability vectors ${\mathbb S}({\cal X})$. Next it will be shown how the  complete characterization of the optimal codes for Problem~\ref{problem1} can be used to  obtain a complete characterization for the above minimax problem, by using  von Neumann's minimax (or minisup) theorem apply. Consider the case when ${\mathbb S}({\cal X})$ is compact (closed and bounded since it is a subset of a finite dimensional space) and convex.  Then, since the set defining the  Kraft inequality in compact and convex, the pay-off  $\alpha \max_{ x\in {\cal X}} l(x) + (1-\alpha) \sum_{ x \in {\cal X}} l(x) p(x)$ is convex and continuous in ${\bf l} \in {\mathbb R}_+^{|{\cal X}|}$ for a  fixed ${\bf p} \in   {\mathbb S}({\cal X})$, and convex and continuous in ${\bf p} \in   {\mathbb S}({\cal X})$ for a fixed ${\bf l} \in {\mathbb R}_+^{|{\cal X}|}$. By  von Neumann's minimax theorem, the minimum over ${\bf l}^* \in {\mathbb R}_+^{|{\cal X}|}$  is interchanged with the maximum over ${\bf p} \in   {\mathbb S}({\cal X})$. Therefore, the solution of Problem~\ref{problem4} is characterized by maximizing over  ${\bf p} \in   {\mathbb S}({\cal X})$, the solution of Problem~\ref{problem1}. On the other hand, if the compactness of the set ${\mathbb S}({\cal X})$ is removed, then the maximization is replaced by supremum and  von Neumann's minsup theorem applies, hence one can interchange the minimum with the supremum utilizing again the solution of   Problem~\ref{problem1}. Hence, the solution to the coding Problem~\ref{problem4} is within our reach and it is based on the solution to Problem~\ref{problem1}.  \\
One may also investigate to what extend von Neumann's minimax theorem holds for the redundancy pay-off (\ref{lm300}); for $\alpha=1$,  ${\mathbb L}_{\alpha}^+({\bf l}+\log {\bf p}, {\bf p})|_{\alpha=1}$,  is  investigated in  \cite{2004:DrmotaSzpankowski,2003:JacquetSzpankowski}.


%
%
%
%
\section{Illustrative Examples}\label{sec:examples}
This section presents two illustrative examples of the optimal codes derived in this paper, with emphasis on the merging rule which partitions the source alphabet ${\cal X}$ into ${\cal U}$ and ${\cal U}^c$ as a function of $\alpha \in [0,1)$.

\subsection{Optimal weights for all $\alpha \in [0,1)$}

\noi Consider binary codewords and a source with $|{\cal X}|=4$ and probability distribution
\begin{align*}
\displaystyle \mathbf{p}=\left(\begin{array}{cccc}
   \frac{8}{15} &  \frac{4}{15} &  \frac{2}{15} &  \frac{1}{15}
\end{array}\right).
\end{align*}
Using Algorithm \ref{charalambous_algorithm} one can find the optimal weight vector $\mathbf{w}_\alpha^{\dagger}$ for different values of $\alpha \in [0,1)$ for which pay-off (\ref{b30}) of Problem~\ref{problem1} is minimized. Computing  $\alpha_1$ via \eqref{akplus1} gives $\alpha_1=1/16$.
For  $\alpha=\alpha_1 =1/16$ the optimal weights are
\begin{align*}
&w_3^{\dagger}(\alpha)=w_4^{\dagger}(\alpha)=(1-\alpha)p_3=\frac{1}{8} \\
&w_2^{\dagger}(\alpha)=(1-\alpha)p_2=\frac{1}{4} \\
&w_1^{\dagger}(\alpha)=(1-\alpha)p_1=\frac{1}{2}
\end{align*}
 In this case, the resulting codeword lengths correspond to the optimal Huffman code.
The weights for all $\alpha \in [0,1)$ can be calculated iteratively by calculating $\alpha_k$ for all $k\in \{ 0, 1, 2, 3\}$ and noting that the weights vary linearly with $\alpha$ (Figure \ref{ex1_p}).

\begin{figure}[H]
\centering
\includegraphics[width=0.6\columnwidth]{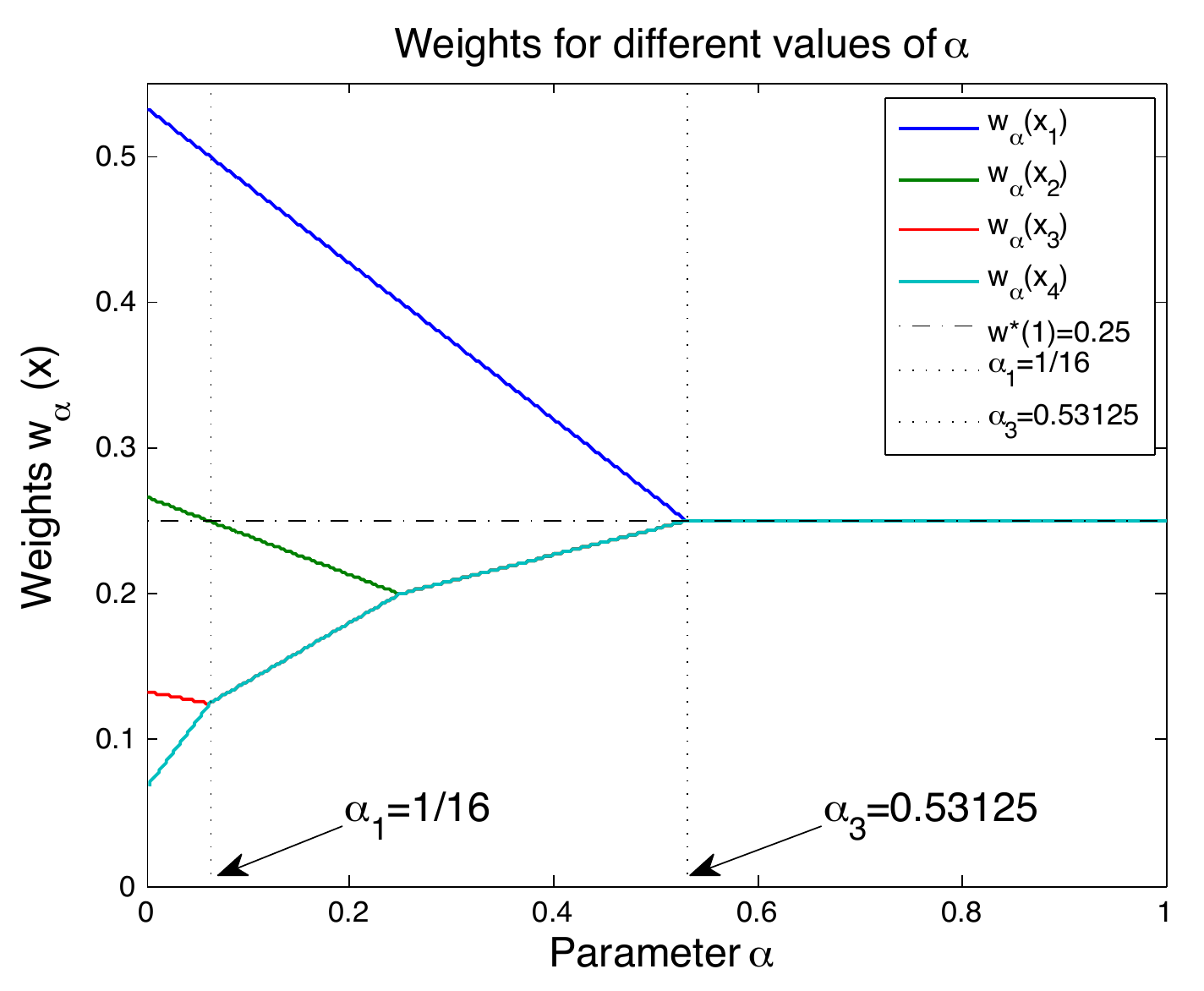}
\caption{A schematic representation of the weights for different values of $\alpha$ when $p=(\frac{8}{15},\frac{4}{15},\frac{2}{15},\frac{1}{15})$.}\label{ex1_p}
\end{figure}

Given the weights, we transformed the problem into a standard average length  coding problem, in which  the optimal codeword lengths can be easily calculated for all $\alpha$'s and they are equal to $ \lceil-\log(w_\alpha(x))\rceil, \forall x \in \mathcal{X}$. The schematic representation of the codeword lengths for $\alpha\in [0,1)$ is shown in Figure \ref{ex1_l}.

\begin{figure}[H]
\centering
\includegraphics[width=0.6\columnwidth]{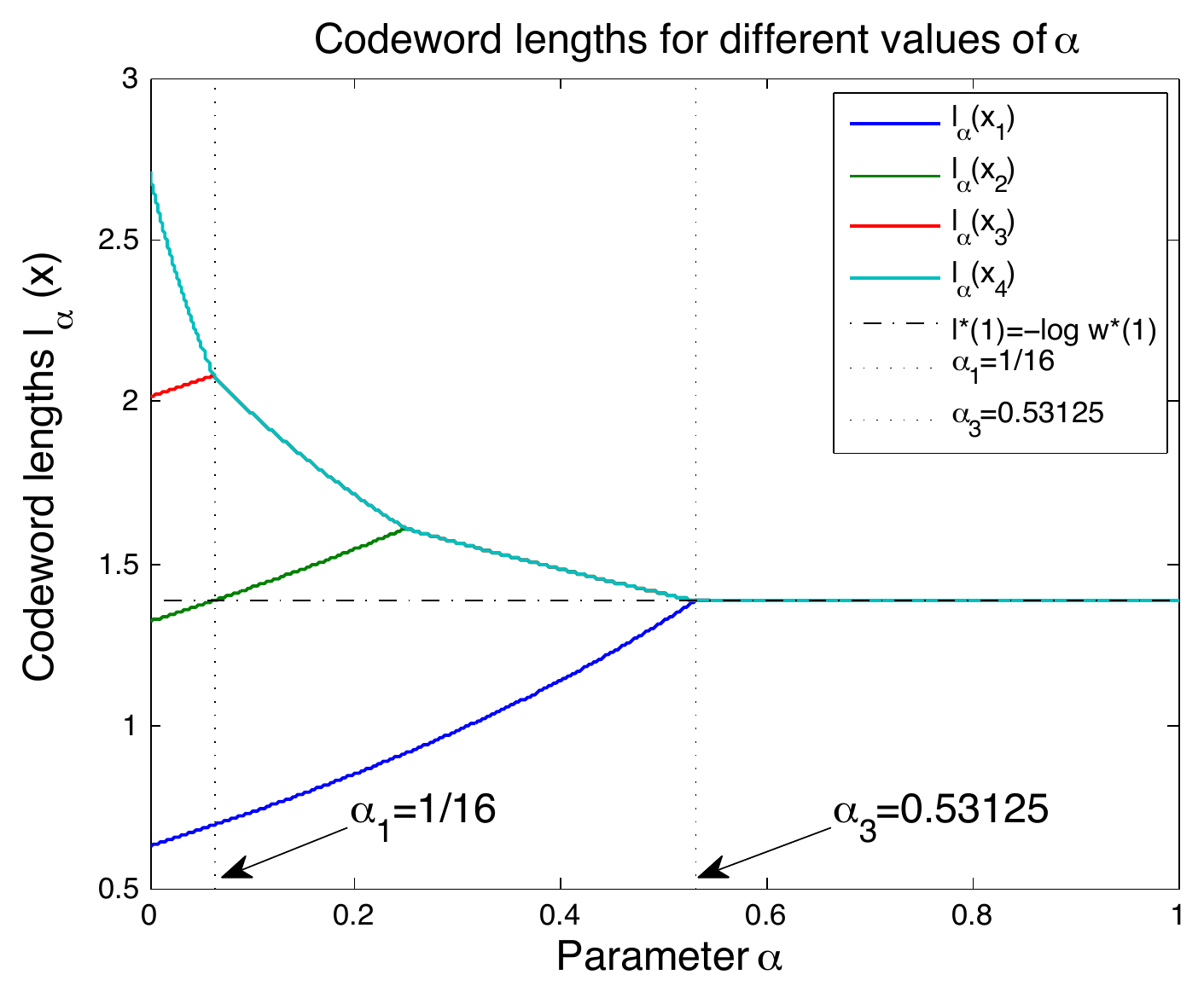}
\caption{A schematic representation of the codeword lengths for different values of $\alpha$ when $p=(\frac{8}{15},\frac{4}{15},\frac{2}{15},\frac{1}{15})$.}\label{ex1_l}
\end{figure}

From Figure \ref{ex1_Lopt} it is verified  that the optimal pay-off function is non-decreasing concave function of $\alpha \in [0,1)$ and at $\alpha_3 = \alpha_{\max} = 0.53125$ the cost function remains unchanged.

\begin{figure}[H]
\centering
\includegraphics[width=0.6\columnwidth]{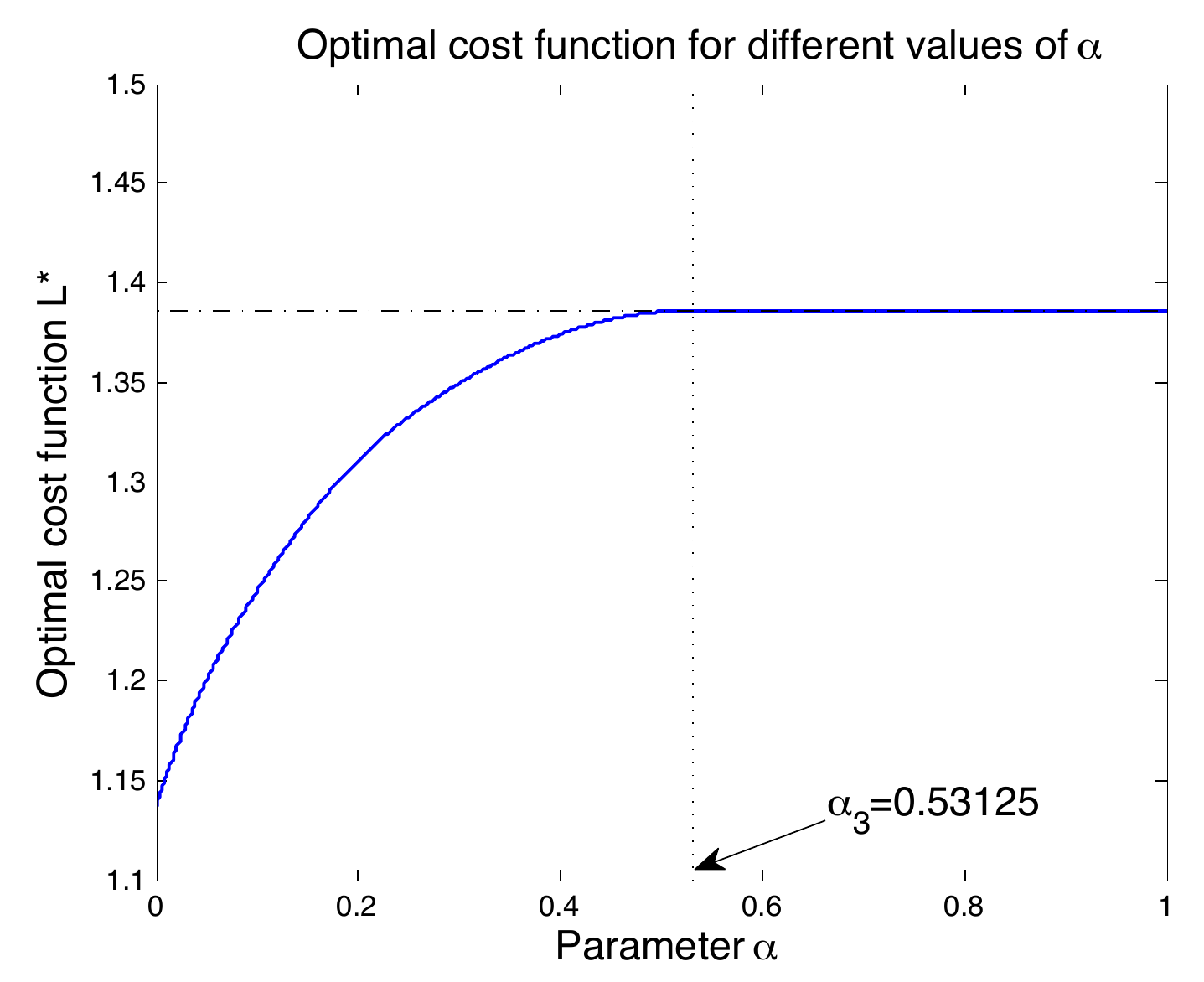}
\caption{A schematic representation of the multiobjective function for different values of $\alpha$ when $p=(\frac{8}{15},\frac{4}{15},\frac{2}{15},\frac{1}{15})$.}\label{ex1_Lopt}
\end{figure}

\subsection{Limited-length coding examples}

\noi Consider binary codewords and a source with $|{\cal X}|=8$ and probability distribution
\begin{align*}
\displaystyle \mathbf{p}=\left(\begin{array}{cccccccc}
  \frac{1}{26} & \frac{1}{26} & \frac{2}{26} &  \frac{2}{26} & \frac{2}{26} & \frac{4}{26} & \frac{5}{26} & \frac{9}{26}
\end{array}\right).
\end{align*}

\noi Using Algorithm \ref{charalambous_limitedLength} one can find the value of $\alpha$ for which the codeword length is  less than or equal to $L_{\lim}$. Hence, the optimal weights $\mathbf{w}^{\dagger}$ and codeword lengths $\mathbf{l}^{\dagger}$ for the given $\alpha$ can be found.\\
Consider, for example, the case  $L_{\lim}=5$; then it can be shown that  $L_{\lim}>-\log (1/26)$ and hence the solution to the problem is the standard Shannon coding with $\alpha=0$. This can also be inferred from Figure \ref{limited2}. Consider the case when the maximum length is $4$ (e.g., $L_{\lim}=4$); then   $\hat{\alpha}=0.0521$ and the optimal lengths are
\begin{align*}
\displaystyle \mathbf{l}^\dagger=\left(\begin{array}{cccccccc}
 1.61 & 2.46 & 2.78 & 3.78 & 3.78 & 3.78 & 4 & 4
\end{array}\right).
\end{align*}

\noi The average codeword length is $2.6355$.

\begin{figure}[H]
\centering
\includegraphics[width=0.7\columnwidth]{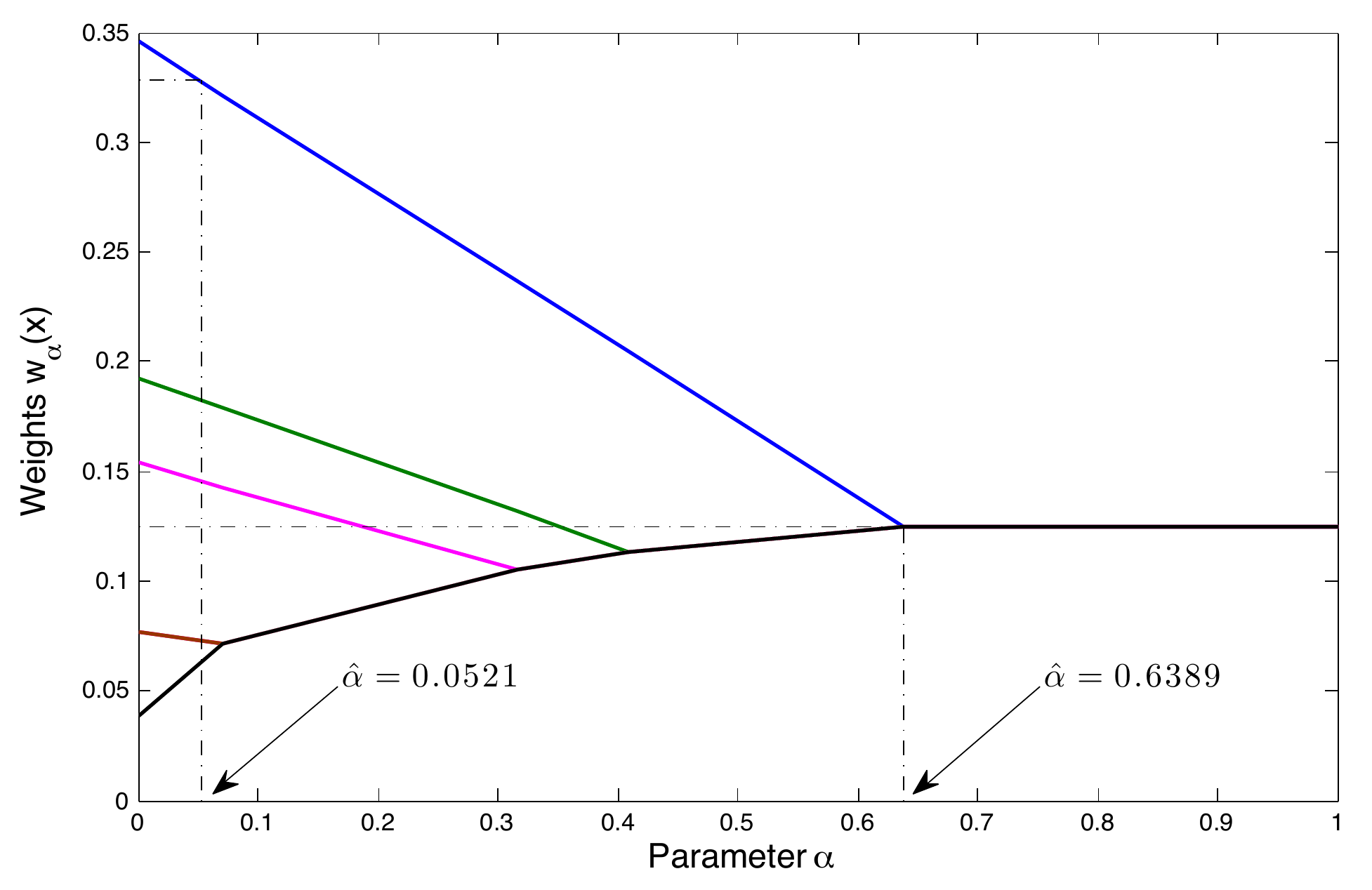}
\caption{A schematic representation of the weights for different values of $\alpha$ when $p=(\frac{1}{26},\frac{1}{26},\frac{2}{26},\frac{2}{26},\frac{2}{26},\frac{4}{26},\frac{5}{26},\frac{9}{26})$.}\label{limited1}
\end{figure}

\noi Consider the case $L_{\lim}=3$; since $|{\cal X}|=8$, there is no compression and all codeword lengths are equal to $3$. In this case,  $\hat{\alpha}=0.6389$, is the minimum $\alpha$ for which there is no compression. This can be seen in Figures \ref{limited1} and \ref{limited2}.

\begin{figure}[H]
\centering
\includegraphics[width=0.7\columnwidth]{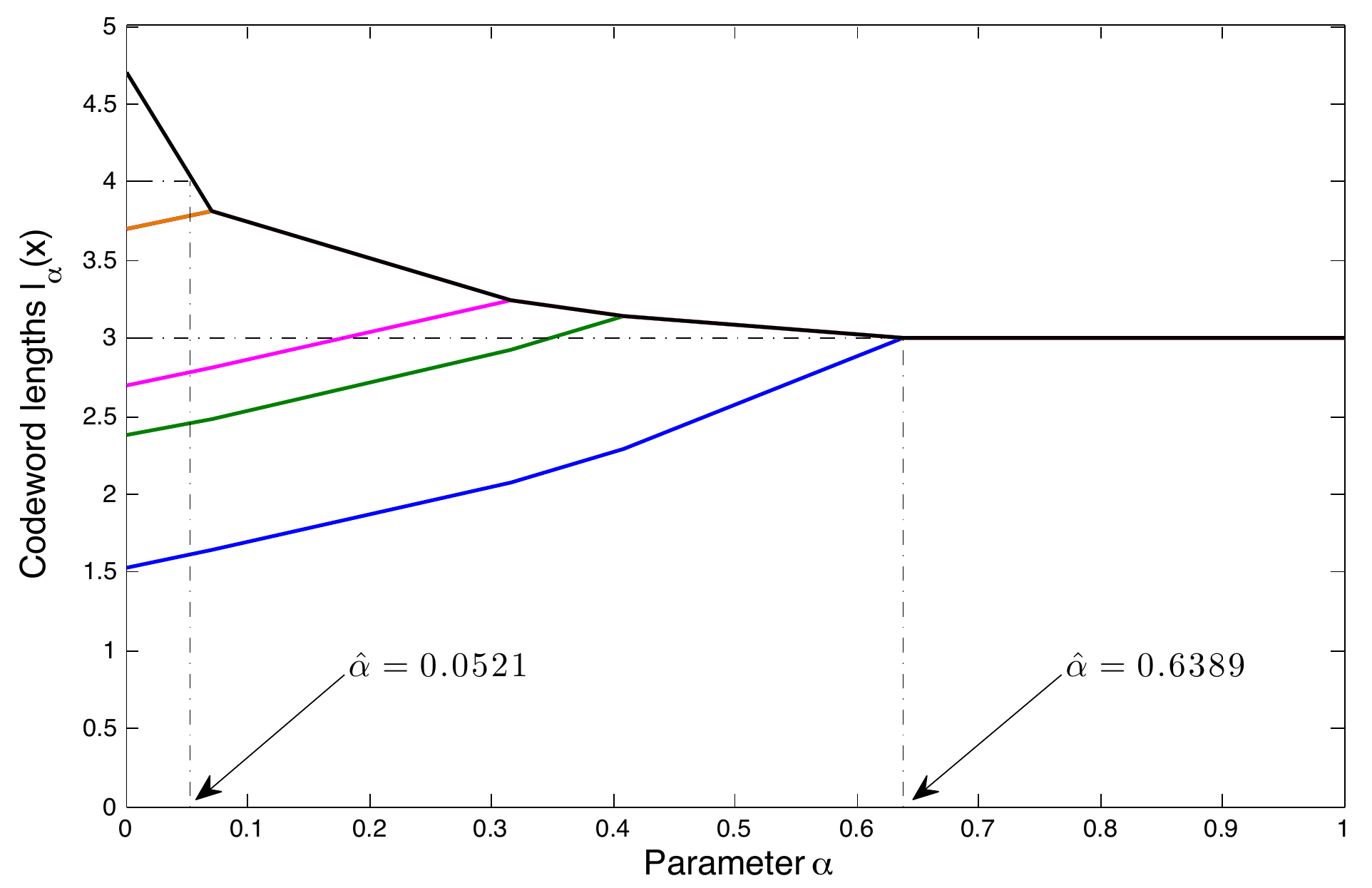}
\caption{A schematic representation of the codeword lengths for different values of $\alpha$ when $p=(\frac{1}{26},\frac{1}{26},\frac{2}{26},\frac{2}{26},\frac{2}{26},\frac{4}{26},\frac{5}{26},\frac{9}{26})$.}\label{limited2}
\end{figure}

Consider the case $L_{\lim}<3$; then there is no $\alpha$ for which the maximum length will be equal $L_{\lim}$.

%
%
%
%
\section{Conclusion and Future Directions}\label{sec:conclusions}

The solution to a lossless coding problem with a pay-off criterion consisting of a convex combination of average and maximum  codeword length is presented. The solution consists of a re-normalization of the initial source probabilities according to a merging  rule. Several properties of the solution are introduced and an algorithm is presented which computes the codeword lengths. The formulation and solution of this problem bridges together an anthology of source coding problems with different pay-offs; relations to problems discussed in the literature are obtained, such as, limited-length coding and, coding with  exponential function of the codeword length. Illustrative examples corroborating the performance of the codes are presented.

The identification of a Huffman-like algorithm which solves the problem using integer-valued codeword lengths is left for future investigation, although it is believed that such an algorithm can be found based on the insight gained in this paper.

%
%
%
%

\appendices

\section{Proofs}

\subsection{Proof of lemma \ref{lemma_prelimin}}\label{proof_lemma_prelimin}

\noi By introducing a real-valued  Lagrange multiplier $\lambda$ associated with the constraint, and using the partition ${\cal X}= {\cal U}\cup {\cal U}^c$,  the augmented pay-off is defined by
\begin{align} \label{Lagrangian_1}
 {\mathbb L}_\alpha({\bf l}, {\bf p}, \lambda) & \tri \alpha l^{*}+(1-\alpha) \sum_{x \in {\cal X}} l(x) p(x) +  \lambda\left(\sum_{x \in {\cal X}}D^{-l(x)}-1\right) \nonumber \\
 &= \sum_{x \in {\cal X}} (\alpha  l^{*} +(1-\alpha)l(x))p(x) +  \lambda\left(\sum_{x \in {\cal X}}D^{-l(x)}-1\right) \nonumber \\
 &= \Big(\alpha  +(1-\alpha) \sum_{x\in {\cal U}} p(x)\Big)l^* + \sum_{ x \in {\cal U}^c} (1-\alpha) p(x)l(x)+ \lambda\left(\sum_{x \in {\cal U}}D^{-l^*} + \sum_{x \in {\cal U}^c}D^{-l(x)}  -1\right).
\end{align}
The augmented pay-off is  a convex and  differentiable function  with respect to ${\bf l}$. Denote the real-valued minimization of \eqref{Lagrangian_1} over ${\bf l}, \lambda$ by  ${\bf l}^\dagger$ and $\lambda^\dagger$. By the Karush-Kuhn-Tucker theorem, the following conditions are necessary and sufficient for  optimality.
 \begin{eqnarray}
\frac{\partial        }{\partial  l(x)}  {\mathbb L}_\alpha({\bf l}, {\bf p}, \lambda) | l_{ {\bf l}={\bf l}^\dagger, \lambda=\lambda^\dagger} &=&0,  \label{lg1} \\
\sum_{x \in {\cal X}}D^{-l^\dagger(x)}  -1\ &\leq& 0,  \label{lg2} \\
\lambda^\dagger \cdot  \left(\sum_{x \in {\cal X}}D^{-l^\dagger(x)}  -1\right)          &=&0, \label{lg3} \\
\lambda^\dagger &\geq& 0. \label{lg4}
\end{eqnarray}
Differentiating with respect to ${\bf l}$, when $x \in {\cal U}$ and  $x \in {\cal U}^c$ the following equations are obtained.
\begin{align}
\frac{\partial        }{\partial  l(x)}  {\mathbb L}_\alpha({\bf l}, {\bf p}, \lambda)  |_{ {\bf l}={\bf l}^\dagger, \lambda=\lambda^\dagger}         &=(1-\alpha)p(x)-\lambda^\dagger D^{-l^\dagger(x)}\log_{e}D=0, \hst  x \in {\cal U}^c  \label{diff1} \\
\frac{\partial        }{\partial  l(x)}  {\mathbb L}_\alpha({\bf l}, {\bf p}, \lambda)  |_{ {\bf l}={\bf l}^\dagger, \lambda=\lambda^\dagger}   &=\alpha \sum_{ x \in {\cal U}^c }p(x) + \sum_{x \in \mathcal{U}}p(x)-\lambda^\dagger |\mathcal{U}| D^{-l^\dagger(x)}\log_{e}D =0, \hst \: x \in {\cal U}. \label{diff2}
\end{align}
When $\lambda^\dagger=0$,   \eqref{diff1} gives  $(1-\alpha)p(x)=0, \forall  x \in {\cal U}^c $. Since $p(x) >0$ then necessarily  $\alpha=1$. This is the case when there is no compression, since ${\cal U}={\cal X}$. For $\alpha \in [0,1)$ then necessarily $\lambda^\dagger >0$. Therefore, by restricting $\alpha \in [0,1)$ then \eqref{diff1}, \eqref{diff2} are equivalent to the following identities.
 \begin{align}\label{li}
D^{-l^\dagger(x)}=\frac{(1-\alpha)p(x)}{\lambda^\dagger \log_{e}D}, \hst  x \in {\cal U}^c,
\end{align}
 \begin{align}\label{lu}
D^{-l^\dagger(x)}=\frac{\alpha \sum_{x \in \mathcal{U}^c}p(x) + \sum_{x \in \mathcal{U}}p(x)}{\lambda^\dagger  |\mathcal{U}| \log_{e}D}, \hst x \in {\cal U}.
\end{align}
Next, $\lambda^\dagger$ is found by substituting \eqref{li} and \eqref{lu} into the Kraft equality to deduce
\begin{align*}
\sum_{x \in {\cal X}   }D^{-l^\dagger(x)} & = \sum_{x \in  \mathcal{U}^c}D^{-l^\dagger(x)}+\sum_{x \in  \mathcal{U}}D^{-l^\dagger(x)} \\
& = \sum_{ x \in {\cal U}^c }\frac{(1-\alpha)p(x)}{\lambda^\dagger  \log_{e}D} + \sum_{x \in  \mathcal{U}}\frac{\alpha \sum_{ x \in {\cal U}^c }p(x) + \sum_{x \in \mathcal{U}}p(x)}{\lambda^\dagger  |\mathcal{U}| \log_{e}D} \\
& =\frac{(1-\alpha) \sum_{ x \in {\cal U}^c }p(x)}{\lambda^\dagger\log_{e}D} +  |\mathcal{U}|  \frac{\alpha \sum_{ x \in {\cal U}^c }p(x) + \sum_{x \in \mathcal{U}}p(x)}{\lambda^\dagger  |\mathcal{U}| \log_{e}D} \\
& = \frac{\sum_{ x \in {\cal U}^c }p(x)+\sum_{x \in \mathcal{U}}p(x)}{\lambda^\dagger\log_{e}D} \\
& = \frac{1}{\lambda^\dagger\log_{e}D} =1 .
\end{align*}
Therefore, $\lambda^\dagger=\frac{1}{\log_{e}D}$.
Substituting $\lambda^\dagger$ into (\ref{li}) and (\ref{lu})  yields
\bes
D^{-l^\dagger(x)}= \left\{ \begin{array}{ll}
(1-\alpha)p(x),  &   x \in {\cal U}^c   \\
\frac{\alpha \sum_{x \in \mathcal{U}^c}p(x) + \sum_{x \in \mathcal{U}}p(x)}{  |\mathcal{U}|}, &  x \in {\cal U}. \end{array} \right.
\ees
Finally, from the previous expression one obtains
\bes
l^\dagger(x) = \left\{ \begin{array}{ll}
-\log{\Big((1-\alpha)p(x)\Big)}, &   x \in {\cal U}^c   \\
-\log{\Big( \frac{\alpha \sum_{x \in \mathcal{U}^c}p(x) + \sum_{x \in \mathcal{U}}p(x)}{  |\mathcal{U}|} \Big)}, & x \in {\cal U} . \end{array} \right.
\ees

\subsection{Proof of Proposition \ref{prop_concavity}}\label{proof_concavity}
Consider the optimal pay-off
\begin{align*}
{\mathbb L}_{\alpha}({\bf l}^\dagger, {\bf p}) \equiv {\mathbb L}({\bf l}^\dagger, {\bf w}_{\alpha}) =  \sum_{ x \in {\cal X}} w_{\alpha}(x) l_{\alpha}^{\dagger}(x) = -\sum_{ x \in {\cal X}} w_{\alpha}(x) \log (w_{\alpha}(x)).
\end{align*}
\noi Differentiating with respect to $\alpha$
\begin{align*}
\frac{ \partial{\mathbb L}({\bf l}^\dagger, {\bf w}_{\alpha}) }{\partial \alpha} &=  -\sum_{ x \in {\cal X}} \frac{\partial \big(w_{\alpha}(x) \log (w_{\alpha}(x))\big)}{\partial \alpha} \\
& =  -\sum_{ x \in {\cal X}} \left( \frac{\partial w_{\alpha}(x)}{\partial \alpha} \log (w_{\alpha}(x)) +   w_{\alpha}(x)\frac{\partial \log (w_{\alpha}(x))}{\partial \alpha}   \right) \\
&=-\sum_{ x \in {\cal X}} \left( w_{\alpha}'(x) \log (w_{\alpha}(x)) +   w_{\alpha}(x)\frac{ w_{\alpha}'(x)}{ w_{\alpha}(x)}\log_D e   \right), \quad \mbox{where } w_{\alpha}'(x)\tri\frac{\partial w_{\alpha}(x)}{\partial \alpha} \\
&=-\sum_{ x \in {\cal X}} w_{\alpha}'(x) \log (w_{\alpha}(x)) -\log_D e \underbrace{\sum_{ x \in {\cal X}}   w_{\alpha}'(x)}_{=0} \\
&=\left(-1+\sum_{ x \in {\cal U}}p(x) \right)\log (w^{*}_{\alpha}) +\sum_{ x \in {\cal U}^c} p(x) \log \big((1-a)p(x)\big)
\end{align*}

\noi Multiplying both sides by $(1-\alpha)$ then
\begin{align*}
(1-\alpha)\frac{ \partial{\mathbb L}({\bf l}^\dagger, {\bf w}_{\alpha}) }{\partial \alpha} &=  -\log (w^{*}_{\alpha})+|{\cal U}|w^{*}_{\alpha}\log (w^{*}_{\alpha})+\sum_{ x \in {\cal U}^c}w_{\alpha}(x)\log \big(w_{\alpha}(x)\big) \\
&=  -\log (w^{*}_{\alpha})+ \sum_{ x \in {\cal X}}w_{\alpha}(x)\log \big(w_{\alpha}(x)\big) \\
&=   \sum_{ x \in {\cal X}}w_{\alpha}(x)\left( -\log (w^{*}_{\alpha})+ \log \big(w_{\alpha}(x)\big)\right) =   \sum_{ x \in {\cal X}}w_{\alpha}(x)\log \left(\frac{w_{\alpha}(x)}{w^{*}_{\alpha}}\right).
\end{align*}

\noi Since ${w_{\alpha}(x)}\geq{w^{*}_{\alpha}}$ then $\log \left(\frac{w_{\alpha}(x)}{w^{*}_{\alpha}}\right)\geq 0$ and therefore, $\frac{ \partial{\mathbb L}({\bf l}^\dagger, {\bf w}_{\alpha}) }{\partial \alpha} \geq 0$. Hence, ${\mathbb L}({\bf l}^\dagger, {\bf w}_\alpha)$ is a non-decreasing function of $\alpha \in [0,1)$. The second derivative of ${\mathbb L}({\bf l}^\dagger, {\bf w}_{\alpha})$ is
\begin{align*}
\frac{ \partial^2{\mathbb L}({\bf l}^\dagger, {\bf w}_{\alpha}) }{\partial \alpha^2} &=\log_D e\left(-1+\sum_{ x \in {\cal U}}p(x) \right)\frac{(w^{*}_{\alpha})'}{w^{*}_{\alpha}} +\log_D e\sum_{ x \in {\cal U}^c} p(x) \left(\frac{-p(x)}{(1-a)p(x)}\right) \\
\frac{1}{\log_D e}\frac{ \partial^2{\mathbb L}({\bf l}^\dagger, {\bf w}_{\alpha}) }{\partial \alpha^2} &= -\frac{\left(1-\sum_{ x \in {\cal U}}p(x) \right)^2}{\alpha +(1-\alpha)\sum_{ x \in {\cal U}}p(x)}-\sum_{ x \in {\cal U}^c} \frac{p(x)}{(1-a)} \leq 0.
\end{align*}
Consequently, ${\mathbb L}({\bf l}^\dagger, {\bf w}_\alpha)$ is a concave non-decreasing function of $\alpha \in [0,1)$.
\noi Note that $\frac{ \partial^2{\mathbb L}({\bf l}^\dagger, {\bf w}_{\alpha}) }{\partial \alpha^2} =\frac{ \partial{\mathbb L}({\bf l}^\dagger, {\bf w}_{\alpha}) }{\partial \alpha} =0$, when $w_{\alpha}(x)=w^{*}_{\alpha}, ~\forall x \in {\cal X}$.


\subsection{Proof of Theorem \ref{mos1}}\label{proof_mos1}

\noi Since the two-parameter pay-off is a convex function and the constraint set is convex, this is  a convex optimization problem. The  augmented pay-off is defined by
\begin{align}
 {\mathbb L}_{t,\alpha}({\bf l}, {\bf p}, \lambda) \tri \alpha  \frac{1}{t}\log \Big(\sum_{ x \in {\cal X}} p(x) D^{t l(x)}\Big)     +(1-\alpha) \sum_{x \in {\cal X}} l(x) p(x) +  \lambda\left(\sum_{x \in {\cal X}}D^{-l(x)}-1\right).  \label{Ln1}
\end{align}
The augmented pay-off is  a convex and  differentiable function  with respect to ${\bf l}$. Denote the real-valued minimization of \eqref{Ln1} over ${\bf l}, \lambda$ by  ${\bf l}^\dagger$ and $\lambda^\dagger$. By the Karush-Kuhn-Tucker theorem, the following conditions are necessary and sufficient for  optimality.
 \begin{eqnarray}
\frac{\partial        }{\partial  l(x)}  {\mathbb L}_{t,\alpha}({\bf l}, {\bf p}, \lambda)  |_{ {\bf l}={\bf l}^\dagger, \lambda=\lambda^\dagger}    &=&0,  \label{lng2} \\
\sum_{x \in {\cal X}}D^{-l^\dagger(x)}   -1\ &\leq& 0,  \label{lng3} \\
\lambda^\dagger \cdot  \left(\sum_{x \in {\cal X}}D^{-l^\dagger(x)}   -1\right)          &=&0, \label{lg3} \\
\lambda^\dagger &\geq& 0. \label{lng4}
\end{eqnarray}

\noi Differentiating (\ref{lng2})  with respect to ${\bf l }$ then

\begin{align}\label{dgiff1}
  \frac{\partial        }{\partial  l(x)}  {\mathbb L}_{t,\alpha}({\bf l}, {\bf p}, \lambda)  |_{ {\bf l}={\bf l}^\dagger, \lambda=\lambda^\dagger}      &   = (1-\alpha) \frac{D^{ t \: l^\dagger(x)} p(x)}{ \sum_{ x \in {\cal X}} p(x) D^{ t \: l^\dagger(x)} } + (1-\alpha )p(x) \nonumber \\
   &-\lambda^\dagger D^{-l^\dagger(x)}\log_{e}D=0, \hst \forall  x \in {\cal X}.
\end{align}

\noi When $\lambda^\dagger=0$, then  \eqref{dgiff1} gives
\begin{align*}
(1-\alpha) \frac{D^{ t \: l^\dagger(x)} p(x)}{ \sum_{ x \in {\cal X}} p(x) D^{ t \: l^\dagger(x)} }  + (1-\alpha )p(x)=0,  \hst  \forall  x \in {\cal X}.
\end{align*}
Summing over ${\cal X}$ then $\alpha +(1-\alpha)=0$, which is impossible, hence  necessarily $\lambda^\dagger >0$.  Consequently, the Kraft inequality holds with equality $\sum_{x \in {\cal X}}D^{-l^\dagger(x)}   =1$. Summing over  ${\cal X}$ on both sides of  (\ref{dgiff1}) gives $\lambda^\dagger=\frac{1}{\log_{e}D}$. Substituting $\lambda^\dagger$ into (\ref{dgiff1}) gives the following set of equations that describe the optimal codeword lengths.

\bea
D^{-l^\dagger(x)}=\alpha \frac{D^{ t \: l^\dagger(x)} p(x)}{ \sum_{ x \in {\cal X}} p(x) D^{ t \: l^\dagger(x)} } + (1-\alpha )p(x), \hst \forall  x \in {\cal X}. \label{gs}
\eea
Consequently, the optimal codeword lengths are given by (\ref{gsl}).\\


\section{Waterfilling-like solution of Problem \ref{problem1}}\label{AppendixB}

\noi The pay-off ${\mathbb L}_{\alpha}^{MO}({\bf l}, {\bf p})$ is a convex combination of the maximum  and the average codeword length. The problem can be expressed as 
\bea
\min_{t}\min_{l} \Big\{ \alpha t + (1-\alpha) \sum_{ x \in {\cal X}} l(x) p(x)\Big\}, \label{b31}
\eea
subject to the Kraft inequality and the constraint $l(x)\leq t$ $\forall x\in {\cal X}$, where $t=\max_{ x\in {\cal X}} l(x)$.

\noi By introducing real-valued  Lagrange multipliers $\lambda(x)$ associated with the constraint $l(x)\leq t$ $\forall x\in {\cal X}$ and a real-valued Lagrange multiplier $\nu$ associate with the Kraft inequality, the augmented pay-off is defined by
\begin{align*} 
 {\mathbb L}_\alpha({\bf l}, {\bf p}, {\bf \lambda}, {\nu}) & \tri \alpha t + (1-\alpha) \sum_{ x \in {\cal X}} l(x) p(x) +  \nu\left(\sum_{x \in {\cal X}}D^{-l(x)}-1\right) +\sum_{x \in {\cal X}}\lambda(x) (l(x)- t)
\end{align*}
The augmented pay-off is  a convex and differentiable function with respect to ${\bf l}$ and $t$. Denote the real-valued minimization over ${\bf l}, t, {\bf \lambda}, \nu$ by  ${\bf l}^\dagger, t^\dagger, {\bf \lambda}^\dagger$ and $\nu^\dagger$. By the Karush-Kuhn-Tucker theorem, the following conditions are necessary and sufficient for  optimality.
 \begin{eqnarray}
\frac{\partial }{\partial  l(x)}  {\mathbb L}_\alpha({\bf l}, {\bf p}, t, {\bf \lambda}, {\nu}) \arrowvert_{ {\bf l}={\bf l}^\dagger, {\bf \lambda}={\bf \lambda}^\dagger, t=t^\dagger, \nu=\nu^\dagger} &=&0,  \label{lg1} \\
\frac{\partial }{\partial  t}  {\mathbb L}_\alpha({\bf l}, {\bf p}, t, {\bf \lambda}, {\nu}) \arrowvert_{ {\bf l}={\bf l}^\dagger, {\bf \lambda}={\bf \lambda}^\dagger, t=t^\dagger, \nu=\nu^\dagger}&=&0,  \label{lg1a} \\
\sum_{x \in {\cal X}}D^{-l^\dagger(x)}  -1\ &\leq& 0,  \label{lg2} \\
\nu^\dagger \cdot  \left(\sum_{x \in {\cal X}}D^{-l^\dagger(x)}  -1\right)          &=&0, \label{lg3} \\
\nu^\dagger &\geq& 0, \label{lg4} \\
 l^\dagger(x)-t &\leq& 0, \forall x\in {\cal X},  \label{lg5} \\
\lambda^\dagger(x) \cdot  \left( l^\dagger(x)-t \right)          &=&0, \forall x\in {\cal X}, \label{lg6} \\
\lambda^\dagger(x) &\geq& 0, \forall x\in {\cal X}. \label{lg7}
\end{eqnarray}
Differentiating with respect to ${\bf l}$, the following equation is obtained:
\begin{align}
\frac{\partial        }{\partial  l(x)}  {\mathbb L}_\alpha({\bf l}, {\bf p}, \lambda, \nu)  |_{ {\bf l}={\bf l}^\dagger, {\bf \lambda}={\bf \lambda}^\dagger, t=t^\dagger, \nu=\nu^\dagger}      &=(1-\alpha)p(x)  -\nu^\dagger D^{-l^\dagger(x)}\log_{e}D+\lambda^\dagger(x)=0,  \label{diff1}
\end{align}
Differentiating with respect to $t$, the following equation is obtained:
\begin{align}
\frac{\partial        }{\partial  t}  {\mathbb L}_\alpha({\bf l}, {\bf p}, \lambda, \nu)  |_{ {\bf l}={\bf l}^\dagger, {\bf \lambda}={\bf \lambda}^\dagger, t=t^\dagger, \nu=\nu^\dagger}  &=  \alpha - \sum_{ x \in {\cal X}}\lambda^\dagger(x) =0. \label{diff2}
\end{align}
When $\nu^\dagger=0$,   \eqref{diff1} gives  $(1-\alpha)p(x)+\lambda^\dagger(x)=0, \forall  x \in {\cal X}$. Since $p(x) >0$ and $\lambda^\dagger(x)\geq 0$ then necessarily  $\alpha=1$. This is the case when there is no compression. For $\alpha \in [0,1)$ then necessarily $\nu^\dagger >0$. 

\noi By restricting $\alpha \in [0,1)$ then \eqref{diff1} and \eqref{diff2} are equivalent to the following identities:
 \begin{align}\label{li}
D^{-l^\dagger(x)}=\frac{(1-\alpha)p(x)+\lambda^\dagger(x)}{\nu^\dagger \log_{e}D}, \hst  x \in {\cal X},
\end{align}
 \begin{align}\label{lambda_x}
\sum_{x \in {\cal X}} \lambda^\dagger(x) =\alpha,
\end{align}
Next, $\nu^\dagger$ is found by substituting \eqref{li} and \eqref{lambda_x} into the Kraft equality to deduce
\begin{align*}
\sum_{x \in {\cal X}   }D^{-l^\dagger(x)} & = \sum_{x \in  \mathcal{X}}\frac{(1-\alpha)p(x)+\lambda^\dagger(x)}{\nu^\dagger  \log_{e}D} \\
& =\frac{(1-\alpha) \sum_{ x \in {\cal X} }p(x)}{\nu^\dagger\log_{e}D} +  \frac{ \sum_{ x \in {\cal X} }\lambda^\dagger(x)}{\nu^\dagger \log_{e}D} \\
& = \frac{(1-\alpha)+ \alpha}{\nu^\dagger\log_{e}D} \\
& = \frac{1}{\nu^\dagger\log_{e}D} =1 .
\end{align*}
Therefore, $\nu^\dagger=\frac{1}{\log_{e}D}$.
Substituting $\nu^\dagger$ into \eqref{li} yields
\bea
D^{-l^\dagger(x)}=(1-\alpha)p(x)+\lambda^\dagger(x),  &   x \in {\cal X},
\eea
Substituting $\lambda^\dagger(x)$ into \eqref{lambda_x} we have
\begin{align}\label{water1}
\sum_{x \in {\cal X}} \big(D^{-l^\dagger(x)}-(1-\alpha)p(x)\big)=\alpha.
\end{align}
Let $w^\dagger(x) \triangleq D^{-l^\dagger(x)}$, i.e., the probabilities that correspond to the codeword lengths $l^\dagger(x)$; also, let  $w^\star \triangleq D^{-t}$. Then, \eqref{water1} can be written as
\begin{align}\label{water2}
\sum_{x \in {\cal X}}\big(w^\dagger(x)-(1-\alpha)p(x)\big)=\alpha.
\end{align}
From the Karush-Kuhn-Tucker conditions \eqref{lg6} and \eqref{lg7} we deduce that for $l(x)<t$, $\lambda(x)=0$ and equation \eqref{water2} becomes
\begin{align}\label{water3}
\sum_{x \in {\cal X}} \big(w^\star-(1-\alpha)p(x)\big)^+=\alpha,
\end{align}
where $(f)^+ = \max(0,f)$. This is the classical waterfilling equation \cite[Section 9.4]{2006:Cover} and $w^\star$ is the water-level chosen, as shown in Figure \ref{waterfilling}.

\begin{figure}[H]
\centering
\includegraphics[width=0.7\columnwidth]{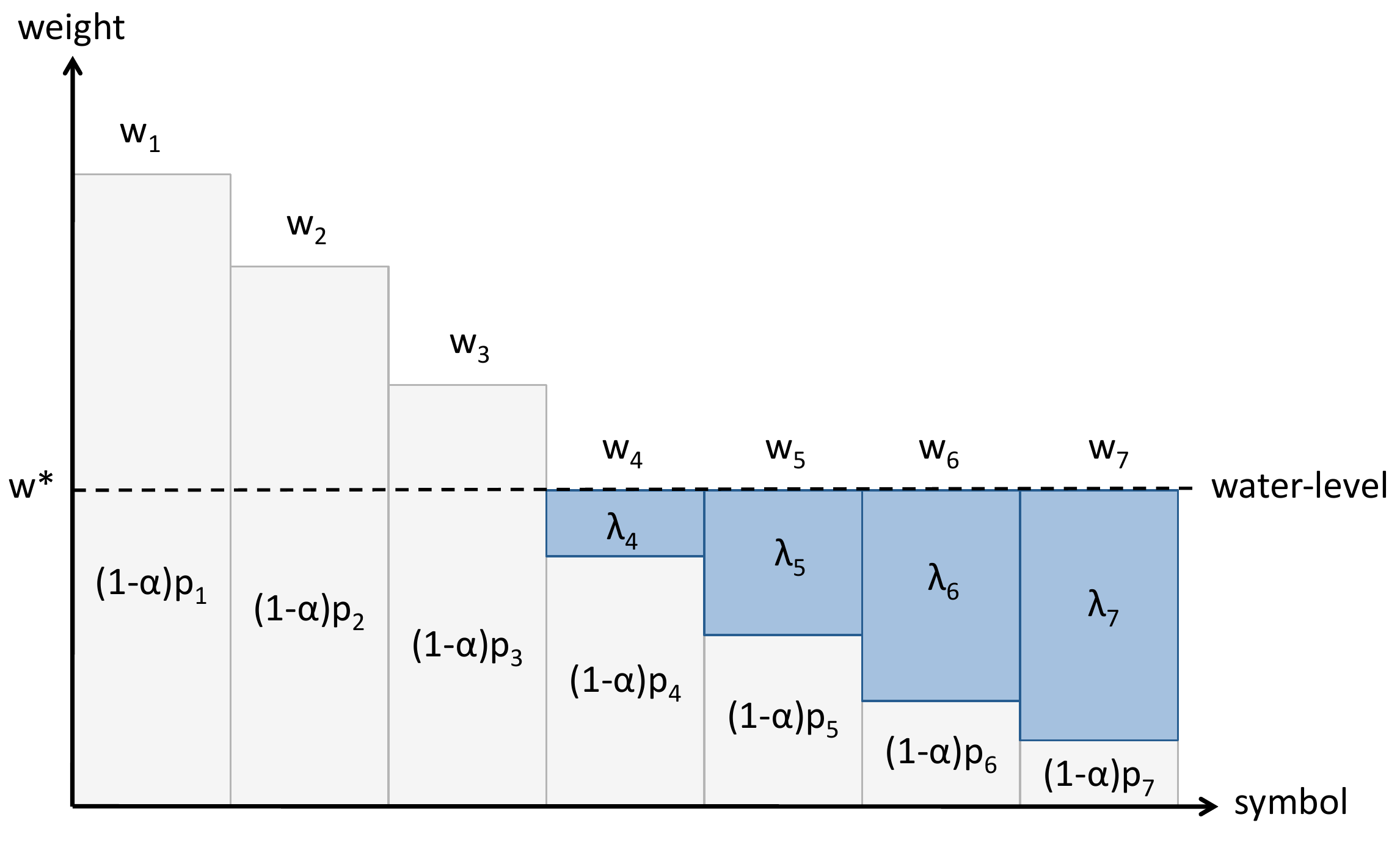}
\caption{Waterfilling solution of the coding problem.}\label{waterfilling}
\end{figure}

\noi For the solution of Problem \ref{problem3a}, for which we consider the limited-length case, $w^\star = D^{-L_{\lim}}$ and equation \eqref{water3} needs to be solved for $\alpha$.


\bibliographystyle{IEEEtran}
\bibliography{bibliografia}
%




\end{document}